\documentclass[a4paper,UKenglish,cleveref, autoref, thm-restate]{lipics-v2021}

\usepackage{xspace}
\usepackage[numbers,sort]{natbib}

\newcommand{\tw}{\mathsf{tw}}

\newcommand{\Oh}{\mathcal{O}}
\newcommand{\poly}{\ensuremath{\mathsf{poly}}}

\newcommand{\sub}{\subseteq}
\newcommand{\sm}{\setminus}

\newcommand{\conpunless}{{NP} $\sub$ {coNP/poly}\xspace}

\newcommand{\zz}{\ensuremath{\mathbb{Z}}}
\newcommand{\rr}{\ensuremath{\mathbb{R}}}
\newcommand{\nn}{\ensuremath{\mathbb{N}}}

\newcommand{\fcal}{\ensuremath{\mathcal{F}}}
\newcommand{\ccal}{\ensuremath{\mathcal{C}}}

\newcommand{\run}[1]{\ensuremath{2^{\Oh(#1)}\cdot n^{\Oh(1)}}}
\newcommand{\polyrun}{\cdot n^{\Oh(1)}}

\def\DEBUG{true}

\ifdefined\DEBUG{}
\newcommand{\mic}[1]{{\color{blue}{#1}}}
\def\rem#1{{\marginpar{\raggedright\scriptsize #1}}}
\newcommand{\micr}[1]{\rem{\textcolor{blue}{\(\bullet \) #1}}}

\newcommand{\verr}[1]{\rem{\textcolor{purple}{\(\bullet \) #1}}}

\newcommand{\karr}[1]{\rem{\textcolor{orange}{\(\bullet \) #1}}}
\else
\newcommand{\mic}[1]{#1}

\newcommand{\micr}[1]{ }
\newcommand{\verr}[1]{ }
\newcommand{\karr}[1]{ }
\fi

\title{Designing Compact ILPs via Fast Witness Verification}

\author{Michał Włodarczyk}{University of Warsaw, Poland}{michal.wloda@gmail.com}{https://orcid.org/0000-0003-0968-8414}{}

\authorrunning{M. Włodarczyk} 

\Copyright{Michał Włodarczyk} 



\category{} 

\relatedversion{An extended abstract of this article has been presented at IPEC 2025~\cite{wlodarczyk25}.
} 

\funding{Supported by Polish National Science Centre SONATA-19 grant 2023/51/D/ST6/00155.}

\nolinenumbers 
\hideLIPIcs

\EventEditors{John Q. Open and Joan R. Access}
\EventNoEds{2}
\EventLongTitle{42nd Conference on Very Important Topics (CVIT 2016)}
\EventShortTitle{CVIT 2016}
\EventAcronym{CVIT}
\EventYear{2016}
\EventDate{December 24--27, 2016}
\EventLocation{Little Whinging, United Kingdom}
\EventLogo{}
\SeriesVolume{42}
\ArticleNo{23}

\newcommand{\defparproblem}[4]{
  \vspace{3mm}
\noindent\fbox{
  \begin{minipage}{0.96\textwidth}
  \begin{tabular*}{\textwidth}{@{\extracolsep{\fill}}lr} \textsc{#1}  & {\bf{Parameter:}} #4 \\ \end{tabular*}
  {\bf{Input:}} #2  \\
  {\bf{Question:}} #3
  \end{minipage}
  }
  \vspace{3mm}
}

\begin{document}

\maketitle

\begin{abstract}
The standard formalization of preprocessing in parameterized complexity is given by kernelization. 
In this work, we depart from this paradigm
and study a different type of preprocessing for problems without polynomial kernels, still aiming at producing instances that are easily solvable in practice.
Specifically, we ask for which parameterized problems an instance $(I,k)$ can be reduced in polynomial time to an integer linear program (ILP) with $\poly(k)$ constraints.

We show that this property coincides with the parameterized complexity class WK[1], previously studied in the context of Turing kernelization lower bounds.
In turn, the class WK[1] enjoys an~elegant characterization in terms of witness verification protocols: a yes-instance should admit a witness of size $\poly(k)$ that can be verified in time $\poly(k)$.
By combining known data structures with new ideas, we design such protocols for several problems, such as {\sc $r$-Way Cut}, {\sc Vertex Multiway Cut}, {\sc Steiner Tree}, and {\sc Minimum Common String Partition}, thus showing that they can be modeled by compact ILPs.
We also present explicit ILP and MILP formulations for {\sc Weighted Vertex Cover} on graphs with small (unweighted) 
vertex cover number.
We believe that these results will provide a background for a systematic study of ILP-oriented preprocessing procedures for parameterized problems.
\end{abstract}

\section{Introduction}

Kernelization~\cite{fomin2019kernelization} provides a rigorous theory for data preprocessing in the following sense: given an instance $I$ with parameter $k$ we would like to transform it in polynomial time into an equivalent instance $(I',k')$ so that $|I'| + k' \le f(k)$ for a slowly growing function $f$.
Historically, kernelization was meant as just another tool for establishing fixed-parameterized tractability~\cite{Downey95} but later it grew into a mature theory on its own, especially after a hardness framework has been established~\cite{Bodlaender09}.
The main question in this paradigm is: which parameterized problems admit polynomial kernels, that is, when can the function $f$ be estimated as polynomial?
From the practical point of view, once the instance size is reduced, it  may become amenable to brute-force or heuristic methods. 

In recent years, empirical observations have provided evidence that 
analyzing preprocessing through the lens of kernelization might be too restrictive.
In the annual PACE challenge~\cite{pace19, pace22, pace21, pace24}, co-organized with the IPEC conference, competitors solve a collection of instances of a certain NP-hard problem related to parameterized complexity.
It turned out that the many successful submissions~\cite{Bannach24pace, bergenthal22pace, boehmer24pace, DirksGRSSS21pace, dobler24pace} do not follow the algorithms from FPT handbooks but instead reduce the problem to (mixed) integer linear programming (ILP) and then employ a specialized solver.
Even though the contestants are not allowed to use the state-of-the-art industrial solvers, already the open source solvers often outperform sophisticated FPT algorithms with provable worst-case guarantees.
However, it is usually challenging to model the problem in question as a concise ILP.

In our opinion, parameterized complexity should seek an explanation of this phenomenon and understand which problems can be handled with
ILP-driven preprocessing.
In this work, we make a first step in this direction and investigate one of the possible ways of formalizing such a~notion of preprocessing.
Even though modeling NP-hard problems as succinct (mixed) ILPs has been a subject of operation research
since 1950s~\cite{chenapplied, dantzig1959linear, floudas2005mixed, goemans1993catalog, Miller60, Nemhauser88, williams2002linear}, 
no systematical study has addressed this question from the perspective of parameterized complexity so far. 

\subparagraph*{What kind of preprocessing?}
Following the design of kernelization theory, we consider polynomial-time procedures that output an equivalent instance (of a possibly different decision problem) that should be well-structured as long as the
initial parameter is not too large. 
While there are many ways to measure the difficulty of an ILP~\cite{JansenK15, Artmann17, Cunningham07, Ganian18, fiorini25, EisenbrandHKKLO25, BrandKO21, Lassota20, CslovjecsekKLPP24, Dvorak17}, 
the most natural are the number of constraints and the number of variables.
In this work, we focus on the first one.

Specifically, we are interested in processing
 an instance $(I,k)$ in polynomial time to rewrite it as an ILP of the form  $(Ax \le b,\, x \in \zz^n_{\ge 0})$, 
where $A \in \zz^{m \times n},\, b \in \zz^m$, and $m$ is polynomially bounded in $k$.
In other words, we would like to reduce solving the problem in question to feasibility check of an ILP with only $\poly(k)$ constraints.
This strategy matches the assumption guiding the kernelization paradigm: when the parameter $k$ is small, then heuristic algorithms should work fast in practice when the search space has ``dimension''\footnote{Usually the ``dimension'' of an ILP refers to the number of variables but the known FPT algorithms~\cite{EisenbrandW20, JansenR23} for solving ILP parameterized by the number of constraints $m$ perform computations in a lattice of dimension~$m$.} $\poly(k)$.
Naturally, it would be desirable to obtain ILP formulations with $m = \Oh(k)$ or $m = \Oh(k^2)$ but we believe that a new framework should first focus on establishing general bounds and more 
fine-grained results will follow.
For comparison, the first work on polynomial kernelization for {\sc Feedback Vertex Set} only guaranteed a~kernel of size $\Oh(k^{11})$~\cite{Burrage06} but it was soon improved to $\Oh(k^2)$~\cite{Thomasse10}.

\subparagraph*{What kind of problems?} Before we delve further into technicalities, let us reflect on for which parameterized problems the question above makes sense.
First of all, we are interested in problems without (known) polynomial kernels.
Secondly, any problem that can be reduced as described above must (1) belong to NP and (2) 
admit an FPT algorithm with running time of the form $2^{\poly(k)}\cdot n^{\Oh(1)}$~\cite{JansenR23}.
Restricting to only such problems allows us to assume $\log n \le \poly(k)$ because otherwise the FPT running time above can be estimated as polynomial in $n$.
Later we will also see another property that must be satisfied: an instance $(I,k)$ should admit an NP-witness of size $\poly(k, \log |I|)$.
Therefore, we should not hope to handle structural parameterizations like treewidth due to known barriers~\cite{DruckerNS16}. 

\subparagraph*{Fast witness verification.}

The question outlined above can be formalized 
in the following way: we ask which parameterized problems admit a {\em polynomial parameter transformation} (PPT) 
to ILP feasibility checking parameterized by the number of constraints. 
It is easy to see that {\sc Set Cover} parameterized by the universe size admits such a reduction and we will present a reduction in the opposite direction as well. 
In turn, the class of problems reducible by PPT to {\sc Set Cover} is known as WK[1].
The WK-hierarchy has been introduced by  Hermelin, Kratsch, Sołtys, Wahlström, and Wu~\cite{HermelinKSWW15} in the context of Turing kernelization.
What is interesting for us, the class WK[1] comprises those parameterized problems that admit an NP-witness of size $\poly(k)$ whose correctness can be verified in time $\poly(k)$.
This elegant characterization has remained merely a curiosity because the main focus of
the work~\cite{HermelinKSWW15} was to prove problem hardness with respect to the WK-hierarchy\footnote{This characterization is implicit in~\cite{HermelinKSWW15} so we provide a proof for completeness}. 
We shall exploit this characterization in a different context. 

\subparagraph*{Related work.}
A more expressive variant of ILP allows one to write lower and upper bounds $\ell_i \le x_i \le u_i$ that do not count to the constraint limit (see \Cref{sec:prelim}).
Rohwedder and Węgrzycki~\cite{Rohwedder25} considered such ILP parameterized by the number of constraints and studied fine-grained equivalences with problems such as {\sc Closest String}, {\sc Discrepancy Minimization}, and {\sc Set Multi-Cover}. 
Kernelization of {\sc ILP Feasibility} has been considered in~\cite{Kratsch13, Kratsch14}.
For more applications of ILP in parameterized complexity, see, e.g.,~\cite{chen2024fpt, Gavenciak22, hermelin2021new, knop2018scheduling}. 

Apart from witness verification protocols studied here, there are other models of nondeterministic computation considered in parameterized complexity, mainly in the context of space limitations~\cite{BodlaenderGNS21, ElberfeldST15, PilipczukWrochna18, Wlodarczyk2024}.

\subsection{Our contribution}

We show that the class WK[1] provides an exact characterization of parameterized problems that can be reduced to ILP with few constraints as long as the entries in the corresponding matrix cannot be too large.
The latter restriction is necessary because unbounded entries allow one to express {\sc Unbounded Subset Sum} (which is NP-hard~\cite{lueker1975two}) with just one equality constraint.
More formally, we consider the feasibility problem of the form $(Ax \le b,\, x \in \zz^n_{\ge 0})$, where $A \in \zz^{m \times n},\, b \in \zz^m$,
and define $\Delta(A)$ to be the maximal absolute value of an entry in the matrix $A$.

\begin{restatable}{theorem}{thmILP}
\label{thm:ilp-wk1}
    {\sc Integer Linear Programming Feasibility} is {\sc WK[1]}-complete when: 
    \begin{itemize}
        \item parameterized by $m + \log(\Delta(A))$,
        \item parameterized by $m$ only, under the assumption that $\Delta(A)=1$ and the vector $b$ is given in unary.
    \end{itemize}
\end{restatable}
\vspace{0.2cm}

Consequently, a parameterized problem admits a {polynomial parameter transformation} to {\sc ILP Feasibility} (under the considered parameterizations) if and only if it belongs to the class WK[1].
Notably, it follows that allowing for polynomially large numbers in an ILP does not lead to any qualitative advantage.
A similar equivalence is known for
the variant with bounded variables~\cite{Rohwedder25} but therein the reduction involves exponentially many variables.
That variant is easily seen to be WK[1]-hard (\Cref{obs:ilp:binary}). 

Hermelin et al.~\cite{HermelinKSWW15} already showed that WK[1] includes problems of local nature,
such as {\sc $k$-Path}, {\sc Connected Vertex Cover}, {\sc Min Ones $d$-SAT}, or {\sc Steiner Tree}, all parameterized by the solution size.
Motivated by \Cref{thm:ilp-wk1}, we show that this class captures many more problems, also those of seemingly non-local nature.
To this end, we employ various data structures to design fast witness verification protocols. 
The definitions of the listed problems and the related  results are discussed in the corresponding sections.

\begin{theorem}\label{thm:list}
    The following parameterized problems belong to the class {\sc WK[1]}:
    \begin{itemize}
        \item {\sc $r$-Way Cut} / solution size,
        \item {\sc Vertex Multiway Cut} / solution size,
        \item {\sc Minimum Common String Partition} / partition size,
        \item {\sc Long Path} / feedback vertex set number,
         \item {\sc Steiner Tree} / number of terminals,
         \item {\sc Optimal Discretization} / solution size.
    \end{itemize}
\end{theorem}

Theorems \ref{thm:ilp-wk1} and \ref{thm:list} imply that all the listed problems can be reduced to ILP with $\poly(k)$ constraints.
In fact, we provide parameterized reductions to {\sc Binary NDTM Halting}, which then reduces to {\sc Set Cover}, which in turn can be readily modeled by a compact ILP.
We remark that the proofs of containment in WK[1] are not technically involved but they exploit dynamic data structures in a novel context.
In our opinion, their value lies in creating a new link between the fields of integer programming, proof protocols, and data~structures.

The results above only provide implicit reductions to ILP where the number constraints is an unspecified polynomial in the parameter.
We complement those with an explicit reduction for {\sc Weighted Vertex Cover / VC}, where the parameter is the vertex cover number of the input graph.
We attain only linearly many constraints, however with a caveat that we require all the variables to be binary.
In addition, we present a mixed integer linear program where the number of integral variables is also linear in the parameter.

\begin{restatable}{theorem}{thmVC}
\label{thm:vc}
    {\sc Weighted Vertex Cover / VC} admits a linear parameter transformation to {\sc Binary ILP Feasibility} parameterized by $m + \log(\Delta(A))$.
    Furthermore, the problem admits a linear parameter transformation to {\sc MILP Feasibility} parametrized by the number of integral variables. 

\end{restatable}

As a corollary of \Cref{thm:vc}, we obtain a very compact ILP formulation for unweighted {\sc Vertex Cover / VC}, where the parameter coincides with the solution size $k$.
Here, both the number of variables and constraints can be bounded as $\Oh(k)$.
Also, $\Delta(A)$ does not exceed $\Oh(k)$. 

\begin{restatable}{corollary}{corVC}
\label{cor:vc}
    {\sc Vertex Cover / VC} admits a linear parameter transformation to {\sc ILP Feasibility} parameterized by $n + m + \Delta(A)$.
\end{restatable}

\subparagraph*{Organization of the paper.}
In \Cref{sec:prelim} we provide background on ILPs and MILPs, together with a discussion on the class WK[1].
\Cref{sec:ilp} is devoted to proving \Cref{thm:ilp-wk1} while
the proof of \Cref{thm:list} is given in Section~\ref{sec:fast1}.
In Sections \ref{sec:vc} and \ref{sec:experiments} we give a closer look at {\sc Weighted Vertex Cover}, first proving \Cref{thm:vc} and then discussing empirical evaluation of this algorithm.
The characterization of the class WK[1] in terms of fast witness verification (\Cref{thm:characterization}) is proven in \Cref{sec:char}.

\section{Preliminaries}\label{sec:prelim}

On default, we measure the running time of algorithms in the standard RAM model with logarithmic word length.
All the summoned algorithmic results are deterministic.
Throughout the paper, we use the variable $n$ to denote the input size during the reductions or the number of variables when analyzing integer programs.
The number of constraints is denoted by $m$.

\subparagraph*{Integer linear programming.}
For $A \in \zz^{m \times n}$ and $b \in \zz^m$, we consider ILP in the {\em standard form without upper bounds}:
$(Ax \le b,\, x \in \zz^n_{\ge 0})$.
Note that it is allowed to write inequalities ``$\ge$'' as well since one can multiply coefficients by -1.
The problem {\sc Integer Linear Programming  (ILP) Feasibility} asks, given $A, b$, whether there exists a vector $x$ satisfying the system above.
One can also consider ILP in the {\em equality form}: $(Ax = b,\, x \in \zz^n_{\ge 0})$ where $A \in \zz^{m \times n}$ and $b \in \zz^m$.
It is easy to transform a system of inequalities $(Ax \le b,\, x \in \zz^n_{\ge 0})$ into the {equality form} by introducing $m$ new non-negative variables and turning each inequality $\sum a_ix_i \le b$ into equality $\sum a_ix_i + z = b$ where $z$ is a fresh variable.
This transformation does not affect the number of constraints and results in matrix $A'$ obtained from $A$ by inserting $m$ columns with only 0 and 1 entries.
We make note of this observation in the following terms.
Recall that $\Delta(A)$ is the maximal absolute value of an entry in the matrix $A$.

\begin{observation}\label{lem:ilp:equality}
A system of inequalities $(Ax \le b,\, x \in \zz^n_{\ge 0})$, $A \in \zz^{m \times n},\, b \in \zz^m$, can be transformed in polynomial time into an equivalent system of equalities $(A'y = b,\, y \in \zz^{n+m}_{\ge 0})$ where $A' \in \zz^{m \times (n+m)}$ and $\Delta(A') = \max(1, \Delta(A))$.
\end{observation}

Another common variant of ILP is the {\em standard form with upper bounds}:
$(Ax \le b,\, x \le u,\, x \in \zz^n_{\ge 0})$ where $u \in  \zz^n_{\ge 0}$.
This form effectively expresses $n$ additional inequalities which are not counted as constraints.
When $u$ is the all-1s-vector, we call the corresponding decision problem {\sc Binary ILP Feasibility}.
All variants admit FPT algorithms parameterized by the number of constraints $m$ and $\Delta(A)$.
On the other hand, parameterization by $m$ alone is unlikely to be FPT even when the input is given in unary~\cite[Theorem 2]{fomin2023optimality}.

\begin{proposition}[{\cite[Theorem 5]{JansenR23}}]
\label{prop:ilp:fpt}
    Feasibility of an ILP in the standard form {\bf without} upper bounds can be decided in time $\Oh(m\cdot \Delta(A))^{2m} \cdot \log (||b||_\infty)$.
\end{proposition}

\begin{proposition}[{\cite[Theorem 4.1]{EisenbrandW20}}]
\label{prop:ilp:fpt-upper}
    Feasibility of an ILP in the standard form {\bf with} upper bounds can be decided in time $(m\cdot \Delta(A))^{\Oh(m^2)} \cdot n + s^{\Oh(1)}$ where $s$ is the input encoding size.
\end{proposition}

In {\em mixed integer linear programming} (MILP) we have two groups of variables: $n$ integral ones and $\ell$ fractional ones.
Let $x[a,b]$ denote the truncation of the vector $x$ to coordinates $[a,b]$.
In a variant with upper bounds we look for a vector $x$ satisfying the system $(Ax \le b,\, x \le u,\, x[1,n] \in \zz^n_{\ge 0},\, x[n+1,n+\ell] \in \rr^{\ell}_{\ge 0})$ where  $A \in \zz^{m \times (n + \ell)}$, $u \in \zz^{n+\ell}$, and $b \in \zz^m$.
When the number $n$ of integral variables is small, one can solve a MILP using an extension of Lenstra's algorithm.
Note that in contrast to parameterization by $m$, now $\Delta(A)$ can be unbounded.

\begin{proposition}[\cite{kannan1987minkowski}]
    Feasibility of a MILP with $n$ integral variables can be determined in time $n^{\Oh(n)}\cdot s^{\Oh(1)}$, where $s$ is the input encoding size.
\end{proposition}

\subparagraph*{Class WK[1].}
A {\em polynomial parameter transformation} (PPT) is a reduction
between parameterized problems $P,Q$
that runs in polynomial time and maps instance $(I,k)$ of $P$ to an equivalent instance $(I',k')$ of $Q$ satisfying
$k' \le \poly(k)$.
The  WK-hierarchy has been introduced
as a hardness hierarchy for problems unlikely to admit a polynomial Turing kernel~\cite{HermelinKSWW15}.
For each $t \ge 1$ the class WK[$t$] comprises problems that are PPT-reducible to {\sc Weighted SAT} over a certain class of formulas, parameterized by the number of 1s in an assignment times $\log n$.
Since we are only interested in the class WK[1], it is more convenient to rely on its equivalent characterization as the PPT-closure of the following WK[1]-complete problem~\cite{HermelinKSWW15}.

\defparproblem{Binary NDTM Halting}
{nondeterministic Turing machine $M$ over a binary alphabet, integer $k$}
{can $M$ halt on the empty string in $k$ steps?}
{$k$}

Another canonical WK[1]-complete problem is {\sc Set Cover}.

\defparproblem{Set Cover}
{set $U$, family of sets $\fcal \sub 2^U$, integer $\ell$}
{is there a subfamily $\fcal' \sub \fcal$ of size $\ell$ such that the union of $\fcal'$ is $U$?}
{$|U|$}

WK[1]-completeness of {\sc Binary NDTM Halting} gives rise to a particularly elegant characterization of WK[1] in terms of witness verification protocols.

\begin{definition}
    A {\em witness verification protocol} for a parameterized problem $L$ is a pair of algorithms $(\cal A,B)$ with the following specification.
    
    Given an instance $(I,k)$ the algorithm $\cal A$ runs in polynomial time and outputs a string $s(I,k) \in \{0,1\}^*$ and integer $\ell(I,k)$.
    The algorithm $\cal B$ is given the concatenation of $s(I,k)$ and a string $w \in \{0,1\}^{\ell(I,k)}$ and outputs a yes/no answer.
    We require that $(I,k) \in L$ if and only if there exists $w$ for which   $\cal B$ answers ``yes''.
    
    For a function $f \colon \nn^2 \to \nn$ we say that the protocol has complexity $f$ if (a) $\ell(I,k) \le f(|I|,k)$ and (b) the running time of $\cal B$ on $s(I,k) + w$ is bounded by $f(|I|,k)$.
\end{definition}

The following characterization was stated implicitly in~\cite[\S 1]{HermelinKSWW15}. 
We provide a~proof in \Cref{sec:char} for completeness.

\begin{restatable}{theorem}{thmCHAR}
\label{thm:characterization}
    A parameterized problem $L$ belongs to the class WK[1] if and only if:
    \begin{enumerate}
        \item $(I,k) \in L$ can be decided in time $2^{\poly(k)}\cdot|I|^{\Oh(1)}$, and
        \item $L$ admits a witness verification protocol of complexity $\poly(k, \log |I|)$. 
     \end{enumerate}
\end{restatable}

\section{WK[1]-completeness of ILP}\label{sec:ilp}

We summon several facts about the structure of solutions to ILP in the equality form.
First we bound the support of a solution using the integral Carath{\'{e}}odory's Theorem.

\begin{proposition}[{\cite[Corollary 5]{EisenbrandS06}}]
\label{lem:ilp:cara}
    Consider an optimization ILP of the form $\min \{c^Tx \colon Ax = b,\, x \in \zz^n_{\ge 0}\}$
    where $A \in \zz^{m \times n}, b \in \zz^m, c \in \zz^n$.
    If this program has a finite optimum, then it has an optimal solution $x^*$ in which the number of non-zero entries is $\Oh(m(\log m + \log(\Delta))$ where $\Delta$ is the maximum absolute value of an entry in $A$ and $c$.
\end{proposition}

Next, we bound the $\ell_1$-norm of a solution.

\begin{proposition}[{\cite[Lemma 3]{JansenR23}}]
\label{prop:ilp:solution-norm}
    If the system $(Ax = b,\, x \in \zz^n_{\ge 0})$ is feasible, it has a solution $x^*$ satisfying $||x^*||_1 \le \Oh(m\cdot \Delta(A))^m \cdot (||b||_\infty + 1)$.
\end{proposition}

At the expense of solving an LP relaxation, we can assume that the entries in the vector $b$ are bounded.
The following proposition is a corollary from the proximity bound by Eisenbrand and Weismantel~\cite{EisenbrandW20}.

\begin{proposition}[{\cite[\S 4.2]{JansenR23}}]
\label{prop:ilp:proximity}
    There is a polynomial-time algorithm that, given $A \in \zz^{m \times n}$ and $b \in \zz^m$, returns a vector $b' \in \zz^m$ such that:
    \begin{enumerate}
        \item  $||b'||_\infty = \Oh(m\cdot \Delta(A))^{m+1}$,
        \item $(Ax = b,\, x \in \zz^n_{\ge 0})$ is feasible if and only if $(Ax = b',\, x \in \zz^n_{\ge 0})$ is feasible.
    \end{enumerate}
\end{proposition}

We are ready to pin down the  complexity status of {\sc ILP Feasibility}.

\thmILP*
\begin{proof}
    We first show that the latter variant is WK[1]-hard. It admits a trivial PPT to the first variant, so then it remains to show that the first one belongs to WK[1].

    We prove WK[1]-hardness by reducing from {\sc Set Cover} parameterized by the universe size.
    This problem is known to be WK[1]-complete~\cite{HermelinKSWW15}.
    For each set $F \in \fcal$ we create a variable $x_F \in \zz_{\ge 0}$.
    For each $u \in U$ we write a constraint $\sum_{F \ni u} x_F \ge 1$, enforcing that $u$ must be covered.
    Then we add a constraint  $\sum_{F \in \fcal} x_F \le \ell$.
    It is easy to see that any integral solution corresponds to a vertex cover of size at most $\ell$.
    This yields a polynomial parameter transformation: the number of constraints is $|U|+1$, $\Delta(A) = 1$, and $||b||_\infty \le |U|$.
    
    Now we show that the first problem belongs to WK[1].
    First we rewrite an ILP into the equality form (\Cref{lem:ilp:equality}).
    By \Cref{thm:characterization} we have to check that the problem is solvable in time $2^{\poly(k)}\cdot s^{\Oh(1)}$ and that it admits a witness verification protocol with complexity $\poly(k,\log s)$, where $s$ stands for the encoding size of the instance. 
    The running time from \Cref{prop:ilp:fpt} can be estimated as $2^{\Oh(k\log k)}\cdot s^{\Oh(1)}$ 
    where $k = m + \Delta(A)$ so it suffices to design a witness verification protocol ($\cal A,B$). 
    The algorithm $\cal A$ preprocesses the ILP using \Cref{prop:ilp:proximity} to ensure that $||b||_\infty = \Oh(m \cdot \Delta(A))^{m+1}$.
    
    We apply \Cref{lem:ilp:cara} for the objective vector $c=0$ so $\Delta$ equals $\Delta(A)$.
    We infer that if the ILP is feasible, then it admits a solution $x^*$ with at most $\Oh(k^2)$ non-zero variables.
    By applying \Cref{prop:ilp:proximity} to the ILP restricted to only these variables we can further assume that $||x^*||_1 \le \Oh(m\cdot \Delta(A))^{2m+1} = 2^{\Oh(k^2\log k)}$.
    Hence $x^*$ can be encoded using  $\Oh(k^2\log s)$ bits for the choice of the non-zero variables and $\Oh(k^4\log k)$ bits for their values.
    This gives an~encoding of the witness supplied to the algorithm~$\cal B$.
    Checking each constraint can be done in time polynomial with respect to the number of non-zero variables and the number of relevant bits, which is $\poly(k,\log s)$.
    This concludes the description of the witness verification protocol.
\end{proof}

By inspecting the reduction from {\sc Set Cover} to {\sc ILP Feasibility}, one can easily see that it also works with additional constraints $x_u \in \{0,1\}$ for every $u \in U$.
It follows that the variant with bounded variables is also WK[1]-hard, however its containment in WK[1] is not known and we find it unlikely (see \Cref{sec:conclusion}).

\begin{observation}\label{obs:ilp:binary}
    {\sc Binary ILP Feasibility} is {\sc WK[1]}-hard when parameterized by $m$, under the assumption that $\Delta(A)=1$ and the vector $b$ is given in unary.
\end{observation}

\section{Fast Witness Verification}\label{sec:fast1}

We establish containment in WK[1] for the problems listed in \Cref{thm:list} using the criterion from
\Cref{thm:characterization}.
The first condition will be always satisfied by summoning a known FPT algorithm and we will focus on designing witness verification protocols. 

\subsection{Cut problems}

In this section we treat {\sc $r$-Way Cut} and {\sc Vertex Multiway Cut}.
We will employ connectivity oracles to quickly verify that a solution candidate forms a cut.

\begin{definition}
    A {\em vertex (resp. edge) failure connectivity oracle} for a graph $G$ and integer $d$ supports two types of operations.
    An {\em update} operation is performed only once and it processes a set $D$ of vertices (resp. edges) from $G$ of size at most $d$.
    In a {\em query} operation, for given vertices $u,v \in G$ the oracle determines whether $u,v$ are connected in $G -D$.
\end{definition}

\begin{proposition}[\cite{LongS22}]
\label{prop:cut:oracle}
There is a deterministic vertex (resp. edge) failure connectivity oracle for undirected graphs with polynomial preprocessing time, update time $\Oh(d^2 \cdot \log^4 n\cdot \log^4 d)$, and worst-case query time $\Oh(d)$.
\end{proposition}

The statement in \cite{LongS22} concerns vertex failure connectivity only but this setting generalizes edge failure connectivity.
To see this, subdivide each edge in $G$ and mimic removal of an edge by removing the corresponding vertex (see also \cite{Holm2001}).

\defparproblem{ $r$-Way Cut}
{undirected graph $G$, integers $r$ and $k$}
{is there a set $X \sub E(G)$ such that $G-X$ has at least $r$ connected components and $|X|\le k$?}
{$k$}

The vertex-deletion variant of the problem is
W[1]-hard already under the parameterization by $k+r$~\cite{MARX2006394}.
Edge deletion is easier and the above stated variant is FPT~\cite{cygan2020randomized}.
However, it does not admit a polynomial kernel unless \conpunless~\cite{Cygan14}.
We will show that {\sc $r$-Way Cut} belongs to WK[1] which, in the light of \Cref{thm:ilp-wk1}, makes the problem amenable to ILP-based preprocessing.

\begin{lemma}\label{lem:rcut:wk1}
    {\sc $r$-Way Cut $\in$ WK[1].}
\end{lemma}
\begin{proof}
    By \Cref{thm:characterization} we need to check that the problem is solvable in time $2^{\poly(k)}n^{\Oh(1)}$ and that it admits a witness verification protocol with complexity $\poly(k,\log n)$.
    The first condition is satisfied due to known $k^{\Oh(k)}\cdot n^{\Oh(1)}$-time algorithm~\cite{cygan2020randomized}.

    Observe that removal of one edge can increase the number of connected component by at most one.
    Let $\ccal(G)$ denote the set of connected components of $G$ and $\ell = |\ccal(G)|$.
    If $\ell + k < r$ then clearly there can be no solution, so let us assume $r - \ell \le k$.
    Suppose that $X \sub E(G)$ is a solution.
    For every $C \in \ccal(G)$ that is being split into $r_C > 1$ parts, there exists vertices $v_1,v_2,\dots, v_{r_C} \in V(C)$ that are pairwise disconnected in $G-X$.
    The number of connected component of $G-X$ equals $\sum_{C \in \ccal(G)} r_C$.
    Since it suffices to split $r - \ell \le k$ components, we may assume that there at most $k$ components $C \in \ccal(G)$ for which $r_C > 1$.
    Therefore, when $X$ is a solution, there must exist $\le 2k$ vertices that can be divided into $t \le k$ groups $V_1,V_2,\dots,V_t$ so that (i) each group $V_i$ is pairwise connected in $G$, (ii) $v \in V_i$ and $u \in V_j$ are not connected in $G$ whenever $i \ne j$, (iii) distinct $u,v \in V_i$ are not connected in $G-X$, and (iv) $\sum_{i=1}^t (|V_i|-1) = r-\ell$.
    This constitutes a certificate that $X$ is a solution.
    
    We now design the witness verification protocol ($\cal A, B$).
    The algorithm $\cal A$ computes the connected components of $G$, assigns vertices in each component a unique label, and
    initializes the data structure from \Cref{prop:cut:oracle}.
    The algorithm $\cal B$ inherits the state of the data structure together with labels marking the connected components and is provided the following witness: a set $X$ of at most $k$ edges and $t \le k$ vertex sets  $V_1,V_2,\dots,V_t$ with $\le 2k$ vertices in total.
    First we check that all the given vertices are pairwise distinct, in time $\Oh(k^2)$.
    Then we check the conditions (i)-(iv).
    The conditions (i)-(ii) can be verified using labels in time $\Oh(k^2)$.
    To handle condition (iii), we first perform the update operation on the edge failure connectivity oracle for set $X$ in time $\Oh(k^3 \cdot \log^4 n)$.
    Then we query the oracle for each pair $u,v \in V_i$ for each $i \in [1,t]$ to check that they are disconnected in $G-X$.
    This takes total time $\Oh(k^3)$.
    Finally, condition (iv) boils down to number additions, which are performed in time $\Oh(k)$.
    This concludes the description of the protocol. 
\end{proof}

We move on to another cut-based problem. 

\defparproblem{Vertex Multiway Cut}
{undirected graph $G$, set $T \sub V(G)$, integer $k$}
{is there a set $X \sub V(G) \sm T$ such that each connected component of $G-X$ has at most one vertex from $T$ and $|X|\le k$?}
{$k$}

The existence of a polynomial kernel for \textsc{Vertex Multiway Cut} remains a major open problem~\cite{KratschW20}.
Positive results include special cases where $|T| = \Oh(1)$~\cite{KratschW20} or where the input graph is planar~\cite{JansenPL21}, and a quasi-polynomial kernel for the edge-deletion variant~\cite{Wahlstrom22}.
In addition, there exist reduction rules that compress the terminal set.

\begin{proposition}[{\cite[Lemma 2.7]{CyganPPW13}}]
\label{prop:multi:reduce}
There is a polynomial-time algorithm that, given an instance $(G,T,k)$ of 
{\sc Vertex Multiway Cut}, outputs an equivalent instance $(G',T',k')$ with $k' \le k$ and $|T'| \le 2k'$.
\end{proposition}

Our witness verification protocol combines this procedure with the vertex failure connectivity oracle.

\begin{lemma}
\label{lem:multi:wk1}
    {\sc Vertex Multiway Cut $\in$ WK[1].}
\end{lemma}
\begin{proof}
    {\sc Vertex Multiway Cut} can be solved in time $2^k \polyrun$~\cite{CyganPPW13} and it remains to construct a witness verification protocol.
    The algorithm $\cal A$ executes the reduction from \Cref{prop:multi:reduce} and so we can assume $|T| \le 2k$.
    Then it initializes the vertex failure connectivity oracle from \Cref{prop:cut:oracle}.
    The algorithm $\cal B$ is given the reduced instance, the state of the oracle, and the witness as simply a solution candidate $X \sub V(G)$, $|X| \le k$.
    First it checks if $X \cap T = \emptyset$ in time $\Oh(k^2)$.
    Then it performs the update operation on the oracle for set $X$, in time $\Oh(k^3 \cdot \log^4 n)$.
    Finally, it asks oracle for each pair $u, v \in T$ whether $u,v$ are disconnected in $G-X$, in total time $\Oh(k^3)$.
    Hence the complexity of the protocol is $\poly(k,\log n)$.
\end{proof}

\subsection{Minimum Common String Partition}

For two strings $x,y$ of equal length, their {\em common string partition} of size $k$ is given by partitions $x = x_1x_2\dots x_k$, $y = y_1y_2\dots y_k$, such that there exists a permutation $f \colon [k] \to [k]$ satisfying $x_i = y_{f(i)}$ for each $i \in [1,k]$.

\defparproblem{Minimum Common String Partition}
{strings $x, y$ of equal length, integer $k$}
{is there a common string partition of $x,y$ of size at most $k$?}
{$k$}

The study of this problem is motivated by applications in comparative
genomics~\cite{BulteauK14}. 
To the best of our knowledge, nothing is known about its kernelization status.
In order to design a witness verification protocol, we will employ the following data structure for maintaining a dynamic collection of strings.

\begin{proposition}[\cite{MehlhornSU97}]
\label{prop:mcsp:dynamic}
    There is a deterministic data structure that supports the following operations in worst-case time 
    $\Oh(\log^2 m\cdot\log^* m)$
    where $m$ is the total number of operations.
    \begin{enumerate}
        \item Create a singleton string with symbol $z$.
        \item Create a new string by concatenating two strings from the collection. 
        \item Create two new strings by splitting a string from the collection at the given index. 
        \item Check if two strings in the collection are equal.
    \end{enumerate}
    Moreover, operations $(2,3)$ keep the host strings in the collection.
\end{proposition}

\begin{lemma}\label{prop:mcsp:wk1}
    {\sc Minimum Common String Partition $\in$ WK[1]}.
\end{lemma}
\begin{proof}
    The problem can be solved in time \run{k^2\log k}~\cite{BulteauK14} and so it suffices to give a witness verification protocol of complexity $\poly(k,\log n)$.
    The algorithm $\cal A$
    builds strings $x,y$ using operations $(1,3)$ of data structure from \Cref{prop:mcsp:dynamic}, then it outputs the computed data structure.

    A witness is a partition $x = x_1x_2\dots x_k$, specified by $k$ indices from the range $[1,n]$, together with a permutation $f \colon [k] \to [k]$.
    The algorithm $\cal B$ inherits the state of the data structure from $\cal A$ and creates strings $x_1,x_2,\dots,x_k$ using the split operation.
    Then it concatenates them in the order given by $f$ and checks if the resulting string is equal to $y$.

    The witness can be encoded using $k\log n + k\log k$ bits.
    The verification process requires $2k+1$ operations on the data structure, each taking time  $\Oh(\log^2 n\cdot\log^* n)$.
\end{proof}

\subsection{Long Path}

A {\em feedback vertex set} (FVS) in a graph $G$ is a set $X \sub V(G)$ for which $G-X$ is acyclic. The FVS number of $G$ is the size of its smallest feedback vertex set.

\defparproblem{Long Path / FVS}
{undirected graph $G$, integer $\ell$}
{is there a path on $\ell$ vertices in $G$?}
{$k=$ FVS number of $G$}

It is known that {\sc Long Path} admits a polynomial kernel when parameterized by the vertex cover number~\cite{bodlaender2013kernel} and in the same paper it was asked whether the same holds for the FVS number parameterization.
To the best of our knowledge, this question remains open since 2013.
An FPT algorithm for {\sc Long Path / FVS} is implied by the following result because treewidth is never larger that the FVS number.

\begin{proposition}[{\cite[Theorem 3.20]{BODLAENDER201586}}]
\label{prop:path:fpt}
    {\sc Long Path} can be solved in time \run{\tw} where $\tw$ denotes the treewidth of the input graph.
\end{proposition}

We remark that this result has been stated under the assumption that a tree decomposition of width $\tw$ is provided in the input but this assumption can be dropped due to known approximation algorithms for treewidth \cite[\S 5]{BODLAENDER201586}.

For a tree $T$ and $u,v \in V(T)$ let $P_T(u,v)$ denote the unique $(u,v)$-path in $T$.
The length of a path $P$ is defined as $|V(P)|$.

\begin{lemma}
\label{lem:path:wk1}
     The problem {\sc Long Path / FVS} belongs to {\sc WK[1]}.
\end{lemma}
\begin{proof}
    The FPT algorithm follows from \Cref{prop:path:fpt} since $\tw(G) \le \mathsf{fvs}(G)$.
    We construct a witness verification protocol ($\cal A,B$).
    The algorithm $\cal A$ first runs 2-approximation for {\sc Feedback Vertex Set}~\cite{bafna19992} so we can assume that a feedback vertex set $X$ of size at most $2k$ is known.
    Next, it computes the connected components of $G-X$ by assigning vertices in each component a unique label.
    Then, on each tree $T$ we compute the following information.
    For each $v_1,v_2 \in V(T)$ we store the length 
    of the path $P_T(v_1,v_2)$.
    For each $v_1,v_2, u_1,u_2 \in V(T)$ we store a bit indicating whether $V(P_T(v_1,v_2)) \cap V(P_T(u_1,u_2)) = \emptyset$.
    The algorithm $\cal A$ outputs all these structures, together with the adjacency matrix of $G$.
    


    Observe that when $P$ is a path in $G$, then it has at most $2k+1$ maximal subpaths $P_1,P_2,\dots,P_t$ in $G-X$.
    The witness is a description of a path $P$ as a sequence of vertices from $X$ and pairs of vertices from $G-X$ that form the endpoints of subpaths $P_i$.
    To verify the witness we need to check (i) that the total length of path is $\ell$, (ii) that the consecutive vertices in the description are adjacent in $G$, and (iii) that the given subpaths are pairwise disjoint.
    To check condition (i) we read the length of each path $P_i$ specified by endpoints $v_1,v_2$.
    This quantity can be retrieved from the precomputed table and
    the total length is then computed using $\Oh(k)$ addition.
    Condition (ii) is tested with the adjacency matrix.
    To check condition (iii), we consider all pairs $i,j \in [1,t]$ and if the endpoints of $P_i,P_j$ belong to the same tree, we determine disjointedness of $P_i,P_j$ using the precomputed table.
    Then we also compare the listed vertices from $X$ to ensure that they are pairwise distinct.
    The witness can be encoded with $\Oh(k\log n)$ bits while the verification takes time $\Oh(k^2)$.
    This yields a witness verification protocol of complexity $\poly(k,\log n)$.
\end{proof}

\subsection{Steiner Tree}

\defparproblem{Steiner Tree / T}
{graph $G$, set $T \sub V(G)$, integer $\ell$}
{is there a tree $X \sub E(G)$ that spans $T$ and $|X| \le \ell$?}
{$k=|T|$}

Hermelin et al.~\cite{HermelinKSWW15} showed that {\sc Steiner Tree} is WK[1]-complete when parameterized by $\ell$.
Since $k \le \ell$ in any non-trivial instance, this implies that the problem is WK[1]-hard when parameterized by $k$ and, furthermore, that it does not admit a polynomial kernel unless \conpunless.

\begin{lemma}
\label{prop:steiner:wk1}
    The problem {\sc Steiner Tree / T} belongs to WK[1].
\end{lemma}
\begin{proof}
    {\sc Steiner Tree} can be  solved in time $3^k\polyrun$~\cite{Dreyfus71} and 
    we present a witness verification protocol.
    The algorithm $\cal A$ computes all the pairwise distances in $G$ and outputs them together with the set $T$.
    The witness supplied to algorithm $\cal B$ is given by a set $Y \sub V(G) \sm T$ of at most $k-1$ vertices and a tree $F$ on the vertex set $T \cup Y$.
    The verification protocol reads the distances in $G$ between each pair $u,v \in T \cup Y$ where $uv \in E(F)$ and checks if their sum is at most $\ell$.
    
    If the verification succeeds then clearly $G$ contains a connected subgraph spanning $T$ on $\le\ell$ edges, so it also contains a tree with this property.
    In the other direction, suppose that there exists a solution tree $X$.
    This tree has at most $k$ leaves so the number of vertices of degree greater than 2 is at most $k-1$.
    The witness guesses such vertices that are not in $T$ and the tree $F$ obtained from $X$ by contracting each vertex $V(X) \sm T$ of degree two to one of its neighbors.
    The total number of edges lying on paths corresponding to $E(F)$ is clearly $|X|$ which bounds the sum of distances read by the algorithm $\cal B$.
    This concludes the correctness proof of the protocol.
    The witness can be encoded using $\Oh(k\log n)$ bits while verification requires $\Oh(k)$ additions of numbers from $[1,n]$.
\end{proof}

\subsection{Optimal Discretization}

Consider sets $W_1,W_2 \sub \mathbb{Q}^2$.
A pair $(X, Y)$ of sets $X, Y \sub \mathbb{Q}$ is called a separation of $(W_1, W_2)$ if for every $(x_1, y_1) \in W_1$ and $(x_2, y_2) \in W_2$ there exists an element of $X$ strictly between $x_1$ and $x_2$ or an
element of $Y$ strictly between $y_1$ and $y_2$.
A geometric interpretation of $(X,Y)$ is a family of axis-parallel lines (vertical line $\ell_x$ for each $x \in X$ and horizontal line $\ell_y$ for each $y \in Y$) so that for every component $C$ of $\rr^2 \sm \bigcup_{z \in X\cup Y} \ell_z$ the closure of $C$ contains either only points from $W_1$ or only points from $W_2$.

\defparproblem{Optimal Discretization}
{sets $W_1,W_2 \sub \mathbb{Q}^2$, integer $k$}
{is there a separation $(X,Y)$ of $(W_1, W_2)$ with $|X| + |Y| \le k$?}
{k}

The geometric interpretation makes the problem attractive from the perspective of learning theory.
This motivated the study of parameterized algorithms for {\sc Optimal Discretization}: it turns out to be FPT~\cite{KratschMMPS21} but only when we consider axis-parallel separating lines~\cite{bonnet2018parameterized}.
To the best of our knowledge, nothing is known about its kernelization status.

For $F \sub \mathbb{Q}$ we define $h(W) = \{\frac{x_1 + x_2}2 \colon\, x_1,x_2 \in F\}$.
For $W \sub \mathbb{Q}^2$ we define sets $h_X(W), h_Y(W) \sub \mathbb{Q}$ as respectively $h(W_X), g(W_Y)$, where $W_X,W_Y$ are projections of $W$ into $X/Y$-coordinates.
It is easy to modify any separation $(X,Y)$ of $(W_1, W_2)$ to only use coordinates that lie exactly in between two coordinates from the input, without affecting the set of pairs that are separated.

\begin{observation}
\label{lem:disc:h}
    There exists a minimum-size separation $(X,Y)$ of $(W_1, W_2)$ satisfying  $X \sub h_X(W_1 \cup W_2),\, Y \sub h_Y(W_1 \cup W_2)$.
\end{observation}

\begin{lemma}\label{lem:disc:wk1}
    {\sc Optimal Discretization $\in$ WK[1].}
\end{lemma}
\begin{proof}
    First we check that {\sc Optimal Discretization} can be solved in time $2^{\Oh(k^2 \log k)}n^{\Oh(1)}$~\cite{KratschMMPS21}.
    The algorithm $\cal A$ of the witness verification protocol considers all tuples $(x_1,x_2,y_1,y_2)$ where $x_1, x_2 \in h_X(W_1 \cup W_2) \cup \{-\infty,\infty\}$, $y_1, y_2 \in h_Y(W_1 \cup W_2) \cup \{-\infty,\infty\}$, $x_1 < x_2$, and $y_1 < y_2$.
    Their number is polynomial in $|W_1| + |W_2|$.
    For a tuple $(x_1,x_2,y_1,y_2)$ it checks whether $[x_1,x_2] \times [y_1,y_2]$ contains at least one point from $W_1$ and from $W_2$.
    If it contains both, then such tuple is marked as bad.
     The algorithm $\cal A$ outputs the set of bad tuples, sorted lexicographically.

     The witness supplied to the algorithm $\cal B$ is a pair $(X,Y)$ where $X \sub h_X(W_1 \cup W_2)$,  $Y \sub h_Y(W_1 \cup W_2)$, and $|X|+|Y|\le k$.
     By \Cref{lem:disc:h} we can assume that there exists an optimal solution of this form.
     Moreover, we require the sets $X,Y$ to be listed in the increasing order of elements. 
     The verification mechanism considers all tuples $(x_1,x_2,y_1,y_2)$ where $x_1,x_2$ are consecutive elements in $X \cup \{-\infty,\infty\}$ and likewise for $y_1,y_2$.
     There are $\Oh(k^2)$ such tuples, each corresponding to a connected component of $\rr^2 \sm \bigcup_{z \in X\cup Y} \ell_z$.
     Then $(X,Y)$ forms a separation for $(W_1,W_2)$ if and only if none of these tuples has been marked as bad.
     One can check if $(x_1,x_2,y_1,y_2)$ is bad by performing binary search in the data structure provided by the algorithm $\cal A$.
     In summary, the witness can be encoded with $\Oh(k\cdot \log n)$ bits whereas verification takes time $\Oh(k^2\cdot \log n)$.
\end{proof}

\section{Weighted Vertex Cover}\label{sec:vc}

The {\em vertex cover number} (VC number) of $G$  is the size of the smallest vertex cover in~$G$.

\defparproblem{Weighted Vertex Cover / VC}
{graph $G$, weight function $w \colon V(G) \to \zz_{\ge 0}$ given in unary, integer $\ell$}
{is there a vertex cover in $G$ of total weight at most $\ell$?}
{VC number of $G$}

The unweighted variant of {\sc Vertex Cover} is probably the most heavily studied problem in kernelization~\cite{chor2005linear, Fellows13, fomin2019kernelization, kratsch2018randomized}.
The weighted variant has a polynomial kernel when parameterized by the number of vertices in the sought solution~\cite{EtscheidKMR17} but it is
unlikely to admit a polynomial kernel when parameterized by the VC number, even when the weights are given in unary~\cite{JansenB13}.
We remark that if the weights were provided in binary, the theorem below would also hold but only when treating the logarithm of the maximal weight as an additional parameter.

\thmVC*
\begin{proof}
We can assume that a vertex cover $Y$ of $G$ of size at most $2k$ is known, due to the classic 2-approximation algorithm.
The problem can be readily solved in time $2^{|Y|} \cdot |V(G)|$ so by the standard argument we can assume that $\log n \le k$, where $n$ is the total input size.
Consider the following mixed integer linear program.

\begin{align*}
	\min & \sum_{v\in V(G)} w(v) \cdot x_v \\
	\forall_{u \in Y} \quad & \sum_{v \in N_G(u)} x_v + \text{deg}_G(u)\cdot x_u \ge \text{deg}_G(u) \\ 
	\forall_{v\in V(G)} \quad & 0 \le x_v \le 1 \\
	\forall_{u\in Y} \quad & x_u \in \zz 
\end{align*}

The number of integral variables, as well as the number of constrains, is $|Y| \le 2k$.
Recall that the bounds $0 \le x_v \le 1$ are not counted as constraints in our setup.

We claim that the optimum of this program equals
the minimum weight vertex cover in $G$.
First, suppose that $S \sub V(G)$ is a vertex cover.
We claim that the binary characteristic vector $x$ of $S$
 satisfies all the constraints.
Consider $u \in Y$.
If $u \in S$ then clearly $\sum_{v \in N_G(u)} x_v + \text{deg}_G(u)\cdot x_u \ge \text{deg}_G(u)$.
If $u \not\in S$ then $N_G(u) \sub S$ and again the inequality holds.

In the other direction, consider a feasible solution $x$. 
We claim that set $S = \{ v\in V(G) \colon x_v = 1\}$ forms a vertex cover.
We need to show that every edge $uw \in E(G)$ is covered by $S$. 
Since $Y$ is also a vertex cover, we can assume w.l.o.g. that $u \in Y$.
Suppose that $x_u < 1$.
Because $x_u$ is an integral variable, this implies $x_u = 0$, and so it must hold that $\sum_{v \in N_G(u)} x_v \ge \text{deg}_G(u)$.
All the variables are upper bounded by 1, which enforces that $x_v = 1$ for every $v \in N_G(u)$, in particular for $v = w$.
Hence the presented MILP indeed models minimum weight vertex cover. 

In order to obtain a feasibility problem we replace the objective function with a constraint $\sum_{v\in V(G)} w(v) \cdot x_v \le \ell$.
This proves the second part of the theorem.
Observe that the entries in the obtained ILP matrix are vertex degrees and vertex weights, which are at most $2^k$ because we assumed $\log n \le k$.
This gives the bound on $\Delta(A)$ and proves the first part of the theorem, by simply treating all the variables as integral.
\end{proof}

It is tempting to try to get rid of the upper bounds $x_v \le 1$ and attain an ILP in the form considered in \Cref{thm:ilp-wk1}.
If we allow for not $\Oh(k)$ but $\poly(k)$ constraints, this is equivalent to proving that {\sc Weighted Vertex Cover / VC} is in WK[1].
A potential witness verification protocol could guess the intersection of a solution $S$ with the vertex cover $Y$ but it is unclear how to verify quickly that the
remaining part of the solution has small weight.

Finally, we justify \Cref{cor:vc} for unweighted {\sc Vertex Cover} parameterized by the solution size $k$.
First, we can run the known kernelization algorithm~\cite{fomin2019kernelization} to transform a given instance $(G,k)$ to $(G',k')$ where $|V(G')| \le 2k$ and $k' \le k$. 
Then the ILP from  \Cref{thm:vc} has at most $2k$ variables, at most $4k$ constraints (including the bounds $x_i \le 1)$, and also $\Delta(A) \le 2k$.

\corVC*

\section{Experiments}
\label{sec:experiments}

This section is devoted to empirical evaluation of the algorithm from~\Cref{thm:vc}.
The experiments were conducted after the conference version of this paper appeared~\cite{wlodarczyk25}.
The code and input data are available on GitHub~\cite{github}.

The goal of the experiments was to assess the utility of the new MILP formulation for {\sc Weighted Vertex Cover} from~\Cref{thm:vc}, compared to the standard ILP, given below.

\begin{align*}
	\min & \sum_{v\in V(G)} w(v) \cdot x_v \\
    \forall_{uv \in E(G)} \quad & x_u + x_v \ge 1 \\
	\forall_{v\in V(G)} \quad &  x_v \ge 0 \\
	\forall_{u\in V(G)} \quad & x_u \in \zz 
\end{align*}

\subsection{Setup}

We have compared the time needed by the ILP solver Gurobi~\cite{gurobi} to solve the same instances of  {\sc Weighted Vertex Cover} on graphs with small vertex cover number.
We implemented both algorithms in Python using {\bf gurobipy} library with Gurobi standard academic license.
In the MILP from~\Cref{thm:vc}, the variables corresponding to the vertices from $Y$ were set as binary, and the remaining ones were continuous from interval $[0,1]$.
Recall than when $k$ denotes the vertex cover number of the input graph, such an MILP has at most $2k$ integral (binary) variables and at most $2k$ non-trivial constraints.

\subsection{Input}
\label{sec:input}

We run our algorithms on random weighted graphs, generated in such a way to ensure that $k$ is bounded.
We believe that random graphs form a good benchmark for testing efficiency of ILP/MILP formulations because 
the popular heuristics used in ILP solvers, such~as branch-and-bound, are known to work fast in expectation on random data~\cite{BorstDHT23, DeyDM20}.
Therefore, one may expect that the heuristics employed by a solver work well even on large inputs, regardless of the particular form of  ILP/MILP, and it should be challenging to improve its performance by the means of combinatorial preprocessing.

Each input graph was constructed as follows.
Given parameters $k$ and $n$, we created a random graph on $n+k$ vertices so that each edge had to touch one the first $k$ vertices.
Each such edge appears independently with probability $1/2$.
Consequently, the vertex cover number is at most $k$ and the number of edges is $\Oh(kn)$.
Then each vertex received its random weight.
The range for sampling weights was set higher for the first $k$ vertices to avoid situation in which they form an optimal solution that is found very quickly by the~solver. 

\subsection{Results}

We have run our experiments for various values of $k$ and $n$, as presented in \Cref{tab:experiments}.
Usually, our algorithm runs faster than the default one by a factor between 6 and 9.
We treat this as an evidence that even the state-of-the-art ILP solver cannot discover the best reduction rules or branching rules without combinatorial insight into a problem.
Even though the MILP from \Cref{thm:vc} is obtained in a relatively simple way, such an optimization is apparently out of reach for a solver that employs problem-oblivious policies.

Finally, let us observe that simply executing the FPT algorithm for {\sc Weighted Vertex Cover} with running time $O(2^k\cdot n)$ would be futile for some the inputs, with $k$ being as high as 100.
In conclusion, it seems that the best approach to handle such instances is to leverage parameterization to produce a compact ILP/MILP and then to take advantage of an optimized solver.

\renewcommand{\arraystretch}{1.4}
\newcolumntype{C}[1]{>{\centering\arraybackslash}p{#1}}

\begin{table}[t]
\centering
\footnotesize
\begin{tabular}{|C{1.5cm}|C{1.5cm}|C{2cm}|C{2cm}|C{2cm}|}
\hline
$k$ & $n$ & MILP from \Cref{thm:vc} & standard ILP & ratio \\ \hline
30 & $10^5$ & 19 & 90 & 4.74 \\ \hline
30 & $10^6$ & 200 & 1471 & 7.36 \\ \hline
40 & $10^4$ & 3 & 15 & 5.00 \\ \hline
40 & $10^5$ & 20 & 169 & 8.45 \\ \hline
40 & $10^6$ & 271 & 1748 & 6.45 \\ \hline
50 & $10^4$ & 3 & 21 & 7.00 \\ \hline
50 & $10^5$ & 24 & 218 & 9.08 \\  \hline
60 & $10^4$ & 3 & 23 & 7.67 \\ \hline
60 & $10^5$ & 29 & 250 & 8.62 \\ \hline
70 & $10^4$ & 4 & 29 & 7.25 \\ \hline
70 & $10^5$ & 35 & 284 & 8.11 \\ \hline
80 & $10^4$ & 4 & 31 & 7.75 \\ \hline
80 & $10^5$ & 36 & 320 & 8.89 \\ \hline
90 & $10^4$ & 4 & 38 & 9.50 \\ \hline
100 & $10^4$ & 4 & 42 & 10.50 \\ \hline
\end{tabular}
\vspace{0.2cm}
\caption{Gurobi runtimes for various instances of {\sc Weighted Vertex Cover} measured in seconds.
The instance parameters $k,n$ are described in \Cref{sec:input}.
In the last column one can read the ratio of runtime needed to find a solution via our MILP to runtime needed to solve the standard ILP. 
}\label{tab:experiments}
\end{table}

\section{Characterization of WK[1]}\label{sec:char}

To prove \Cref{thm:characterization} we will reduce problems with fast witness verification protocols to {\sc Binary NDTM Halting}.
Since the running time is measured in the RAM model, we need to translate the protocol to short computation of a Turing machine.
However, the standard translation methods do not preserve sublinear running times because a regular Turing machine spends a lot of time moving the head along the tape.
To overcome this issue, we will consider Turing machines that can operate their heads more efficiently.

\begin{definition}[Random-access Turing Machine (RTM) \cite{GurevichS89}]
An RTM is a Turing machine with three linear tapes over binary alphabet, called the {\em main tape} $M$,
the {\em address tape} $A$ and the {\em auxiliary tape} $X$, such that the head of the main tape (the main
head) is always in the cell whose number is the contents of the address tape in binary. An instruction
for an RTM has the form
\[
(p, \alpha_M, \alpha_A, \alpha_X) \to (q, \beta_M, \beta_A, \beta_X, \gamma_A, \gamma_X)
\]
and means the following: If the control state is $p$ and the symbols in the observed cells on
the three tapes are binary digits $\alpha_M, \alpha_A, \alpha_X$ respectively, then print binary digits $\beta_M, \beta_A, \beta_X$
in the respective cells, move the address head to the left (resp. to the right) if $\gamma_A = -1$
(resp. $\gamma_A = +1$), move the auxiliary tape head with respect to $\gamma_X$, and go to control state $q$.
\end{definition}

\begin{proposition}[\cite{GurevichS89}]
\label{prop:char:simulate}
    Any computation in the RAM model with logarithmic word length on $n$-length input can be simulated by a RTM 
    with $\mathsf{polylog}(n)$
     overhead in the running time and with lengths of the address and auxiliary tapes bounded by $\Oh(\log n)$. 
\end{proposition}

\thmCHAR*
\begin{proof}
    If $L \in$ WK[1] then it admits a PPT to the WK[1]-complete problem {\sc Set Cover} parameterized by the universe size.
    It is easy to see that {\sc Set Cover} enjoy both properties and they are preserved by PPT.

    In the other direction, we will show that any problem satisfying both conditions admits a PPT to {\sc Binary NDTM Halting}, which is WK[1]-complete~\cite{HermelinKSWW15}.
    Consider the verification algorithm $\cal B$ that is given the output of the algorithm $\cal A$ together with a witness of length $\poly(k,\log n)$ and verifies it in time $\poly(k,\log n)$.
    By \Cref{prop:char:simulate} there exists an RTM ${\cal B}_1$ that performs the same computation as $\cal B$ within time $\poly(k,\log n)$ and so that the address and auxiliary tapes populate only $\Oh(\log n)$ first cells.

    We will modify ${\cal B}_1$ so that the main tape $M$ is read-only.
    Let $\cal T$ be a standard (i.e., not random access) Turing machine with the following specification. 
    The input word $y$ is a list of pairs $(a_i, b_i)$ where $a_i$ is a binary string of length $s$ and $b_i \in \{0,1\}$ and a string $a$.
    The machine finds the pair occurrence $(a_i,b_i)$ with $a_i = a$ and finishes in one of three states, reporting that either $b_i = 0$, $b_i = 1$, or that such a pair does not exit, and moving the head to the initial position.
    Clearly, there exists such a machine $\cal T$ that runs in  time $\Oh(|y|^c)$ for some $c\in \nn$.

    We now create an RTM ${\cal B}_2$ with additional auxiliary tape $Y$ whose states are pairs $(q_B,q_T)$ where $q_B$ is a state if ${\cal B}_1$ and $q_T$ is a state of ${\cal T}$.
    The machine ${\cal B}_2$ simulates ${\cal B}_1$ in a following way: instead of writing a symbol to the tape $M$ it appends a pair $(a,b)$ to the tape $Y$ where $a$ is the current content of the address tape and $b$ is the symbol meant to be written to $M$.
    After reading a symbol from $M$ it performs computation of machine $\cal T$ on the tape $Y$.
    After this subroutine, the state of $\cal T$ informs what symbol would be read by  ${\cal B}_1$ at that time.
    If the machine ${\cal B}_1$ runs for $t$ steps, then the input length for the subroutine is $\Oh(t\cdot \log n)$,
    hence this modification incurs running time overhead of $\Oh((t\cdot \log n)^c)$.
    Clearly ${\cal B}_2$ accepts the same strings as ${\cal B}_1$ and its running time is still of the form  $\poly(k,\log n)$.

    We are ready to describe the reduction from $L$ to {\sc Binary NDTM Halting}.
    Given an instance $(I,k)$ we first check if $\log n > p(k)$ where $p(k)$ is the exponent of the FPT algorithm from condition one.
    If yes, the running time $2^{p(k)}\cdot|I|^{\Oh(1)}$ becomes polynomial in $|I|$ and so we can further assume  $\log n \le p(k)$.
    The reduction first runs the algorithm $\cal A$ and receives outputs $s$.
    We showed that the existence of the algorithm $\cal B$ implies the existence of the RTM ${\cal B}_2$ with the following properties: (1) ${\cal B}_2$ works on tapes $M,A,X,Y$ with read-only access to $M$, (2) the length of $A,X$ is bounded by $\Oh(\log n)$, (3) ${\cal B}_2$ finishes within time $\ell \le \poly(k,\log n) = \poly(k)$ in input $x+w$ where $w$ is a witness of prescribed length.
    We can easily modify this machine to have an additional auxiliary tape $W$ so that initially $M$ stores $x$ and $W$ stores $w$; let us call this machine ${\cal B}_3$.
    The input to {\sc Binary NDTM Halting} is a nondeterministic Turing machine ${\cal H}_s$ whose states are triples $(q,a,x)$ where $q$ is a state of ${\cal B}_3$ and $a,x$ are binary strings of length $\Oh(\log n)$.
    The instructions are given by mimicking the instructions ${\cal B}_3$ assuming that $s$ is the concent of tape $M$.
    Note that ${\cal H}_s$ has description size $n^{\Oh(1)}$.
    The machine ${\cal H}_s$ first nondeterministically guesses the word $w$ on tape $W$ and then simulates ${\cal B}_3$ on the input $(s,w)$.
    Observe that ${\cal H}_s$ can simulate ${\cal B}_3$ because the latter machine is read-only with respect to the tape $M$ storing $s$.
    The two tapes $Y,W$ can be easily merged into a single one while keeping the running time polynomial in $k$.
    If the simulation would accept the input $(s,w)$ before the $\ell$-th step then ${\cal H}_s$ is ordered to halt and otherwise it loops.
    Hence the machine can halt if and only if $(I,k)$ admits a witness $w$ that would be accepted by the protocol.
    Since $\ell = \poly(k)$, this yields a PPT to {\sc Binary NDTM Halting}. 
\end{proof}

\section{Conclusion}
\label{sec:conclusion}

We have presented a characterization of parameterized problems reducible to ILP with few constraints and classified several new problems into this category through witness verification protocols.
We find this connection surprising and inspiring for exploring new directions in parameterized complexity. 
Can we further extend this classification?
In particular we ask whether the following problems belong to WK[1]: {\sc Weighted Vertex Cover / VC}, {\sc Connected FVS}, {\sc Directed FVS}, and {\sc Vertex Planarization} (the last three parameterized by the solution size).

Our classification concerns ILPs in the standard form without upper bounds.
Allowing for upper bounds on variables gives a greater expressive power when the number of constraints is a parameter (see \Cref{obs:ilp:binary} and also~\cite{Rohwedder25}).
It would be interesting to also characterize parameterized problems reducible to upper-bounded ILP with few constraints. 
However, the location of such a feasibility problem in the parameterized complexity landscape remains unclear because we do not know whether it admits short NP-witnesses~\cite{Wlodarczyk2024}.
Yet another direction is to describe problems that are reducible to ILP with few variables.

Containment in WK[1] only ensures that a  problem can be modeled by ILP with $\poly(k)$ constraints.
Even though the degree of this polynomial can be retraced from the chain of reductions, it is unlikely to be optimal.
As the natural next step towards bridging the gap between theory and practice, we would like to optimize the polynomial degree, by designing explicit ILP formulation with 
$\Oh(k^c)$ constraints, where $c$ is as small as possible, ideally one.
A good starting point would be to consider
problems with simple proofs of containment in WK[1], such as
{\sc $k$-Path} or {\sc Connected Vertex Cover}~\cite{HermelinKSWW15}. 

\bibliography{bib}

\begin{thebibliography}{10}

\bibitem{gurobi}
{Gurobi Optimization}.
\newblock URL: \url{http://www.gurobi.com}.

\bibitem{Artmann17}
Stephan Artmann, Robert Weismantel, and Rico Zenklusen.
\newblock A strongly polynomial algorithm for bimodular integer linear programming.
\newblock In {\em Proceedings of the 49th Annual ACM SIGACT Symposium on Theory of Computing}, STOC 2017, page 1206–1219, New York, NY, USA, 2017. Association for Computing Machinery.
\newblock \href {https://doi.org/10.1145/3055399.3055473} {\path{doi:10.1145/3055399.3055473}}.

\bibitem{bafna19992}
Vineet Bafna, Piotr Berman, and Toshihiro Fujito.
\newblock A 2-approximation algorithm for the undirected feedback vertex set problem.
\newblock {\em SIAM Journal on Discrete Mathematics}, 12(3):289--297, 1999.
\newblock \href {https://doi.org/10.1137/S0895480196305124} {\path{doi:10.1137/S0895480196305124}}.

\bibitem{Bannach24pace}
Max Bannach, Florian Chudigiewitsch, Kim-Manuel Klein, and Marcel Wien\"{o}bst.
\newblock {PACE Solver Description: UzL Exact Solver for One-Sided Crossing Minimization}.
\newblock In \'{E}douard Bonnet and Pawe{\l} Rz\k{a}\.{z}ewski, editors, {\em 19th International Symposium on Parameterized and Exact Computation (IPEC 2024)}, volume 321 of {\em Leibniz International Proceedings in Informatics (LIPIcs)}, pages 28:1--28:4, Dagstuhl, Germany, 2024. Schloss Dagstuhl -- Leibniz-Zentrum f{\"u}r Informatik.
\newblock \href {https://doi.org/10.4230/LIPIcs.IPEC.2024.28} {\path{doi:10.4230/LIPIcs.IPEC.2024.28}}.

\bibitem{bergenthal22pace}
Moritz Bergenthal, Jona Dirks, Thorben Freese, Jakob Gahde, Enna Gerhard, Mario Grobler, and Sebastian Siebertz.
\newblock {PACE Solver Description: GraPA-JAVA}.
\newblock In Holger Dell and Jesper Nederlof, editors, {\em 17th International Symposium on Parameterized and Exact Computation (IPEC 2022)}, volume 249 of {\em Leibniz International Proceedings in Informatics (LIPIcs)}, pages 30:1--30:4, Dagstuhl, Germany, 2022. Schloss Dagstuhl -- Leibniz-Zentrum f{\"u}r Informatik.
\newblock \href {https://doi.org/10.4230/LIPIcs.IPEC.2022.30} {\path{doi:10.4230/LIPIcs.IPEC.2022.30}}.

\bibitem{BODLAENDER201586}
Hans~L. Bodlaender, Marek Cygan, Stefan Kratsch, and Jesper Nederlof.
\newblock Deterministic single exponential time algorithms for connectivity problems parameterized by treewidth.
\newblock {\em Information and Computation}, 243:86--111, 2015.
\newblock 40th International Colloquium on Automata, Languages and Programming (ICALP 2013).
\newblock \href {https://doi.org/10.1016/j.ic.2014.12.008} {\path{doi:10.1016/j.ic.2014.12.008}}.

\bibitem{Bodlaender09}
Hans~L. Bodlaender, Rodney~G. Downey, Michael~R. Fellows, and Danny Hermelin.
\newblock On problems without polynomial kernels.
\newblock {\em Journal of Computer and System Sciences}, 75(8):423--434, 2009.
\newblock \href {https://doi.org/10.1016/j.jcss.2009.04.001} {\path{doi:10.1016/j.jcss.2009.04.001}}.

\bibitem{BodlaenderGNS21}
Hans~L. Bodlaender, Carla Groenland, Jesper Nederlof, and C{\'{e}}line M.~F. Swennenhuis.
\newblock Parameterized problems complete for nondeterministic {FPT} time and logarithmic space.
\newblock In {\em 62nd {IEEE} Annual Symposium on Foundations of Computer Science, {FOCS} 2021, Denver, CO, USA, February 7-10, 2022}, pages 193--204. {IEEE}, 2021.
\newblock \href {https://doi.org/10.1109/FOCS52979.2021.00027} {\path{doi:10.1109/FOCS52979.2021.00027}}.

\bibitem{bodlaender2013kernel}
Hans~L Bodlaender, Bart~MP Jansen, and Stefan Kratsch.
\newblock Kernel bounds for path and cycle problems.
\newblock {\em Theoretical Computer Science}, 511:117--136, 2013.
\newblock \href {https://doi.org/10.1016/j.tcs.2012.09.006} {\path{doi:10.1016/j.tcs.2012.09.006}}.

\bibitem{boehmer24pace}
Kimon Boehmer, Lukas~Lee George, Fanny Hauser, and Jesse Palarus.
\newblock {PACE Solver Description: Arcee}.
\newblock In {\em 19th International Symposium on Parameterized and Exact Computation (IPEC 2024)}, volume 321 of {\em Leibniz International Proceedings in Informatics (LIPIcs)}, pages 33:1--33:4, Dagstuhl, Germany, 2024. Schloss Dagstuhl -- Leibniz-Zentrum f{\"u}r Informatik.
\newblock \href {https://doi.org/10.4230/LIPIcs.IPEC.2024.33} {\path{doi:10.4230/LIPIcs.IPEC.2024.33}}.

\bibitem{bonnet2018parameterized}
\'{E}douard Bonnet, Panos Giannopoulos, and Michael Lampis.
\newblock {On the Parameterized Complexity of Red-Blue Points Separation}.
\newblock In Daniel Lokshtanov and Naomi Nishimura, editors, {\em 12th International Symposium on Parameterized and Exact Computation (IPEC 2017)}, volume~89 of {\em Leibniz International Proceedings in Informatics (LIPIcs)}, pages 8:1--8:13, Dagstuhl, Germany, 2018. Schloss Dagstuhl -- Leibniz-Zentrum f{\"u}r Informatik.
\newblock \href {https://doi.org/10.4230/LIPIcs.IPEC.2017.8} {\path{doi:10.4230/LIPIcs.IPEC.2017.8}}.

\bibitem{BorstDHT23}
Sander Borst, Daniel Dadush, Sophie Huiberts, and Samarth Tiwari.
\newblock On the integrality gap of binary integer programs with gaussian data.
\newblock {\em Math. Program.}, 197(2):1221--1263, 2023.
\newblock \href {https://doi.org/10.1007/S10107-022-01828-1} {\path{doi:10.1007/S10107-022-01828-1}}.

\bibitem{BrandKO21}
Cornelius Brand, Martin Kouteck{\'{y}}, and Sebastian Ordyniak.
\newblock Parameterized algorithms for {MILPs} with small treedepth.
\newblock In {\em Thirty-Fifth {AAAI} Conference on Artificial Intelligence, {AAAI} 2021, Thirty-Third Conference on Innovative Applications of Artificial Intelligence, {IAAI} 2021, The Eleventh Symposium on Educational Advances in Artificial Intelligence, {EAAI} 2021, Virtual Event, February 2-9, 2021}, pages 12249--12257. {AAAI} Press, 2021.
\newblock \href {https://doi.org/10.1609/AAAI.V35I14.17454} {\path{doi:10.1609/AAAI.V35I14.17454}}.

\bibitem{BulteauK14}
Laurent Bulteau and Christian Komusiewicz.
\newblock Minimum common string partition parameterized by partition size is fixed-parameter tractable.
\newblock In Chandra Chekuri, editor, {\em Proceedings of the 25th Annual {ACM-SIAM} Symposium on Discrete Algorithms, {SODA} 2014, Portland, Oregon, USA, January 5-7, 2014}, pages 102--121. {SIAM}, 2014.
\newblock \href {https://doi.org/10.1137/1.9781611973402.8} {\path{doi:10.1137/1.9781611973402.8}}.

\bibitem{Burrage06}
Kevin Burrage, Vladimir Estivill-Castro, Michael Fellows, Michael Langston, Shev Mac, and Frances Rosamond.
\newblock The undirected feedback vertex set problem has a poly(k) kernel.
\newblock In Hans~L. Bodlaender and Michael~A. Langston, editors, {\em Parameterized and Exact Computation}, pages 192--202, Berlin, Heidelberg, 2006. Springer Berlin Heidelberg.
\newblock \href {https://doi.org/10.1007/11847250_18} {\path{doi:10.1007/11847250_18}}.

\bibitem{chenapplied}
Der-San Chen, Robert~G Batson, and Yu~Dang.
\newblock {\em Applied Integer Programming: Modeling and Solution}.
\newblock Wiley Online Library, 2009.
\newblock \href {https://doi.org/10.1002/9781118166000} {\path{doi:10.1002/9781118166000}}.

\bibitem{chen2024fpt}
Hua Chen, Lin Chen, and Guochuan Zhang.
\newblock Fpt algorithms for a special block-structured integer program with applications in scheduling.
\newblock {\em Mathematical Programming}, 208(1):463--496, 2024.
\newblock \href {https://doi.org/10.1007/s10951-017-0550-0} {\path{doi:10.1007/s10951-017-0550-0}}.

\bibitem{chor2005linear}
Benny Chor, Mike Fellows, and David Juedes.
\newblock Linear kernels in linear time, or how to save k colors in {$O(n^2)$} steps.
\newblock In {\em Graph-Theoretic Concepts in Computer Science: 30th International Workshop, WG 2004, Bad Honnef, Germany, June 21-23, 2004. Revised Papers 30}, pages 257--269. Springer, 2005.
\newblock \href {https://doi.org/10.1007/978-3-540-30559-0_22} {\path{doi:10.1007/978-3-540-30559-0_22}}.

\bibitem{CslovjecsekKLPP24}
Jana Cslovjecsek, Martin Kouteck{\'{y}}, Alexandra Lassota, Michał Pilipczuk, and Adam Polak.
\newblock Parameterized algorithms for block-structured integer programs with large entries.
\newblock In David~P. Woodruff, editor, {\em Proceedings of the 2024 {ACM-SIAM} Symposium on Discrete Algorithms, {SODA} 2024, Alexandria, VA, USA, January 7-10, 2024}, pages 740--751. {SIAM}, 2024.
\newblock \href {https://doi.org/10.1137/1.9781611977912.29} {\path{doi:10.1137/1.9781611977912.29}}.

\bibitem{Cunningham07}
William~H. Cunningham and Jim Geelen.
\newblock On integer programming and the branch-width of the constraint matrix.
\newblock In Matteo Fischetti and David~P. Williamson, editors, {\em Integer Programming and Combinatorial Optimization}, pages 158--166, Berlin, Heidelberg, 2007. Springer Berlin Heidelberg.
\newblock \href {https://doi.org/10.1007/978-3-540-72792-7_13} {\path{doi:10.1007/978-3-540-72792-7_13}}.

\bibitem{cygan2020randomized}
Marek Cygan, Pawe{\l} Komosa, Daniel Lokshtanov, Marcin Pilipczuk, Micha{\l} Pilipczuk, Saket Saurabh, and Magnus Wahlstr{\"o}m.
\newblock Randomized contractions meet lean decompositions.
\newblock {\em ACM Transactions on Algorithms (TALG)}, 17(1):1--30, 2020.
\newblock \href {https://doi.org/10.1145/3426738} {\path{doi:10.1145/3426738}}.

\bibitem{Cygan14}
Marek Cygan, Stefan Kratsch, Marcin Pilipczuk, Micha\l{} Pilipczuk, and Magnus Wahlstr\"{o}m.
\newblock Clique cover and graph separation: New incompressibility results.
\newblock {\em ACM Trans. Comput. Theory}, 6(2), May 2014.
\newblock \href {https://doi.org/10.1145/2594439} {\path{doi:10.1145/2594439}}.

\bibitem{CyganPPW13}
Marek Cygan, Marcin Pilipczuk, Michal Pilipczuk, and Jakub~Onufry Wojtaszczyk.
\newblock On multiway cut parameterized above lower bounds.
\newblock {\em {ACM} Trans. Comput. Theory}, 5(1):3:1--3:11, 2013.
\newblock \href {https://doi.org/10.1145/2462896.2462899} {\path{doi:10.1145/2462896.2462899}}.

\bibitem{dantzig1959linear}
George~Bernard Dantzig, Delbert~R Fulkerson, and Selmer~Martin Johnson.
\newblock On a linear-programming, combinatorial approach to the traveling-salesman problem.
\newblock {\em Operations Research}, 7(1):58--66, 1959.
\newblock \href {https://doi.org/10.1287/opre.7.1.58} {\path{doi:10.1287/opre.7.1.58}}.

\bibitem{DeyDM20}
Santanu~S. Dey, Yatharth Dubey, and Marco Molinaro.
\newblock {\em Branch-and-Bound Solves Random Binary IPs in Polytime}, pages 579--591.
\newblock \href {https://doi.org/10.1137/1.9781611976465.35} {\path{doi:10.1137/1.9781611976465.35}}.

\bibitem{DirksGRSSS21pace}
Jona Dirks, Mario Grobler, Roman Rabinovich, Yannik Schnaubelt, Sebastian Siebertz, and Maximilian Sonneborn.
\newblock {PACE} solver description: {PACA-JAVA}.
\newblock In Petr~A. Golovach and Meirav Zehavi, editors, {\em 16th International Symposium on Parameterized and Exact Computation, {IPEC} 2021, September 8-10, 2021, Lisbon, Portugal}, volume 214 of {\em LIPIcs}, pages 30:1--30:4. Schloss Dagstuhl - Leibniz-Zentrum f{\"{u}}r Informatik, 2021.
\newblock \href {https://doi.org/10.4230/LIPICS.IPEC.2021.30} {\path{doi:10.4230/LIPICS.IPEC.2021.30}}.

\bibitem{dobler24pace}
Alexander Dobler.
\newblock {PACE Solver Description: CRGone}.
\newblock In \'{E}douard Bonnet and Pawe{\l} Rz\k{a}\.{z}ewski, editors, {\em 19th International Symposium on Parameterized and Exact Computation (IPEC 2024)}, volume 321 of {\em Leibniz International Proceedings in Informatics (LIPIcs)}, pages 29:1--29:4, Dagstuhl, Germany, 2024. Schloss Dagstuhl -- Leibniz-Zentrum f{\"u}r Informatik.
\newblock \href {https://doi.org/10.4230/LIPIcs.IPEC.2024.29} {\path{doi:10.4230/LIPIcs.IPEC.2024.29}}.

\bibitem{Downey95}
Rodney~G. Downey and Michael~R. Fellows.
\newblock Parameterized computational feasibility.
\newblock In Peter Clote and Jeffrey~B. Remmel, editors, {\em Feasible Mathematics II}, pages 219--244, Boston, MA, 1995. Birkh{\"a}user Boston.
\newblock \href {https://doi.org/10.1007/978-1-4612-2566-9_7} {\path{doi:10.1007/978-1-4612-2566-9_7}}.

\bibitem{Dreyfus71}
S.~E. Dreyfus and R.~A. Wagner.
\newblock The steiner problem in graphs.
\newblock {\em Networks}, 1(3):195--207, 1971.
\newblock \href {https://doi.org/10.1002/net.3230010302} {\path{doi:10.1002/net.3230010302}}.

\bibitem{DruckerNS16}
Andrew Drucker, Jesper Nederlof, and Rahul Santhanam.
\newblock {Exponential Time Paradigms Through the Polynomial Time Lens}.
\newblock In Piotr Sankowski and Christos Zaroliagis, editors, {\em 24th Annual European Symposium on Algorithms (ESA 2016)}, volume~57 of {\em Leibniz International Proceedings in Informatics (LIPIcs)}, pages 36:1--36:14, Dagstuhl, Germany, 2016. Schloss Dagstuhl -- Leibniz-Zentrum f{\"u}r Informatik.
\newblock \href {https://doi.org/10.4230/LIPIcs.ESA.2016.36} {\path{doi:10.4230/LIPIcs.ESA.2016.36}}.

\bibitem{Dvorak17}
Pavel Dvo\v{r}\'{a}k, Eduard Eiben, Robert Ganian, Du\v{s}an Knop, and Sebastian Ordyniak.
\newblock Solving integer linear programs with a small number of global variables and constraints.
\newblock In {\em Proceedings of the 26th International Joint Conference on Artificial Intelligence}, IJCAI'17, page 607–613. AAAI Press, 2017.
\newblock \href {https://doi.org/10.5555/3171642.3171730} {\path{doi:10.5555/3171642.3171730}}.

\bibitem{pace19}
M.~Ayaz Dzulfikar, Johannes~K. Fichte, and Markus Hecher.
\newblock {The PACE 2019 Parameterized Algorithms and Computational Experiments Challenge: The Fourth Iteration}.
\newblock In Bart M.~P. Jansen and Jan~Arne Telle, editors, {\em 14th International Symposium on Parameterized and Exact Computation (IPEC 2019)}, volume 148 of {\em Leibniz International Proceedings in Informatics (LIPIcs)}, pages 25:1--25:23, Dagstuhl, Germany, 2019. Schloss Dagstuhl -- Leibniz-Zentrum f{\"u}r Informatik.
\newblock \href {https://doi.org/10.4230/LIPIcs.IPEC.2019.25} {\path{doi:10.4230/LIPIcs.IPEC.2019.25}}.

\bibitem{EisenbrandHKKLO25}
Friedrich Eisenbrand, Christoph Hunkenschr{\"{o}}der, Kim{-}Manuel Klein, Martin Kouteck{\'{y}}, Asaf Levin, and Shmuel Onn.
\newblock Sparse integer programming is fixed-parameter tractable.
\newblock {\em Math. Oper. Res.}, 50(3):2141--2156, 2025.
\newblock \href {https://doi.org/10.1287/MOOR.2023.0162} {\path{doi:10.1287/MOOR.2023.0162}}.

\bibitem{EisenbrandS06}
Friedrich Eisenbrand and Gennady Shmonin.
\newblock Carath{\'{e}}odory bounds for integer cones.
\newblock {\em Oper. Res. Lett.}, 34(5):564--568, 2006.
\newblock \href {https://doi.org/10.1016/J.ORL.2005.09.008} {\path{doi:10.1016/J.ORL.2005.09.008}}.

\bibitem{EisenbrandW20}
Friedrich Eisenbrand and Robert Weismantel.
\newblock Proximity results and faster algorithms for integer programming using the {Steinitz} lemma.
\newblock {\em {ACM} Trans. Algorithms}, 16(1):5:1--5:14, 2020.
\newblock \href {https://doi.org/10.1145/3340322} {\path{doi:10.1145/3340322}}.

\bibitem{ElberfeldST15}
Michael Elberfeld, Christoph Stockhusen, and Till Tantau.
\newblock On the space and circuit complexity of parameterized problems: Classes and completeness.
\newblock {\em Algorithmica}, 71(3):661--701, 2015.
\newblock \href {https://doi.org/10.1007/S00453-014-9944-Y} {\path{doi:10.1007/S00453-014-9944-Y}}.

\bibitem{EtscheidKMR17}
Michael Etscheid, Stefan Kratsch, Matthias Mnich, and Heiko R{\"{o}}glin.
\newblock Polynomial kernels for weighted problems.
\newblock {\em J. Comput. Syst. Sci.}, 84:1--10, 2017.
\newblock \href {https://doi.org/10.1016/j.jcss.2016.06.004} {\path{doi:10.1016/j.jcss.2016.06.004}}.

\bibitem{Fellows13}
Michael~R. Fellows, Bart~M.P. Jansen, and Frances Rosamond.
\newblock Towards fully multivariate algorithmics: Parameter ecology and the deconstruction of computational complexity.
\newblock {\em European Journal of Combinatorics}, 34(3):541--566, 2013.
\newblock Combinatorial Algorithms and Complexity.
\newblock \href {https://doi.org/10.1016/j.ejc.2012.04.008} {\path{doi:10.1016/j.ejc.2012.04.008}}.

\bibitem{fiorini25}
Samuel Fiorini, Gwena\"{e}l Joret, Stefan Weltge, and Yelena Yuditsky.
\newblock Integer programs with bounded subdeterminants and two nonzeros per row.
\newblock {\em J. ACM}, 72(1), January 2025.
\newblock \href {https://doi.org/10.1145/3695985} {\path{doi:10.1145/3695985}}.

\bibitem{floudas2005mixed}
Christodoulos~A Floudas and Xiaoxia Lin.
\newblock Mixed integer linear programming in process scheduling: Modeling, algorithms, and applications.
\newblock {\em Annals of Operations Research}, 139:131--162, 2005.
\newblock \href {https://doi.org/10.1007/s10479-005-3446-x} {\path{doi:10.1007/s10479-005-3446-x}}.

\bibitem{fomin2019kernelization}
Fedor Fomin, Daniel Lokshtanov, Saket Saurabh, and Meirav Zehavi.
\newblock {\em Kernelization: theory of parameterized preprocessing}.
\newblock Cambridge University Press, 2019.
\newblock \href {https://doi.org/10.1017/9781107415157} {\path{doi:10.1017/9781107415157}}.

\bibitem{fomin2023optimality}
Fedor~V Fomin, Fahad Panolan, Maadapuzhi~Sridharan Ramanujan, and Saket Saurabh.
\newblock On the optimality of pseudo-polynomial algorithms for integer programming.
\newblock {\em Mathematical Programming}, 198(1):561--593, 2023.
\newblock \href {https://doi.org/10.1007/s10107-022-01783-x} {\path{doi:10.1007/s10107-022-01783-x}}.

\bibitem{Ganian18}
Robert Ganian and Sebastian Ordyniak.
\newblock The complexity landscape of decompositional parameters for ilp.
\newblock {\em Artificial Intelligence}, 257:61--71, 2018.
\newblock \href {https://doi.org/10.1016/j.artint.2017.12.006} {\path{doi:10.1016/j.artint.2017.12.006}}.

\bibitem{Gavenciak22}
Tomáš Gavenčiak, Martin Koutecký, and Dušan Knop.
\newblock Integer programming in parameterized complexity: Five miniatures.
\newblock {\em Discrete Optimization}, 44:100596, 2022.
\newblock Optimization and Discrete Geometry.
\newblock \href {https://doi.org/10.1016/j.disopt.2020.100596} {\path{doi:10.1016/j.disopt.2020.100596}}.

\bibitem{goemans1993catalog}
Michel~X Goemans and Young-Soo Myung.
\newblock A catalog of steiner tree formulations.
\newblock {\em Networks}, 23(1):19--28, 1993.
\newblock \href {https://doi.org/10.1002/net.3230230104} {\path{doi:10.1002/net.3230230104}}.

\bibitem{pace22}
Ernestine Gro{\ss}mann, Tobias Heuer, Christian Schulz, and Darren Strash.
\newblock {The PACE 2022 Parameterized Algorithms and Computational Experiments Challenge: Directed Feedback Vertex Set}.
\newblock In Holger Dell and Jesper Nederlof, editors, {\em 17th International Symposium on Parameterized and Exact Computation (IPEC 2022)}, volume 249 of {\em Leibniz International Proceedings in Informatics (LIPIcs)}, pages 26:1--26:18, Dagstuhl, Germany, 2022. Schloss Dagstuhl -- Leibniz-Zentrum f{\"u}r Informatik.
\newblock \href {https://doi.org/10.4230/LIPIcs.IPEC.2022.26} {\path{doi:10.4230/LIPIcs.IPEC.2022.26}}.

\bibitem{GurevichS89}
Yuri Gurevich and Saharon Shelah.
\newblock Nearly linear time.
\newblock In Albert~R. Meyer and Michael~A. Taitslin, editors, {\em Logic at Botik '89, Symposium on Logical Foundations of Computer Science, Pereslav-Zalessky, USSR, July 3-8, 1989, Proceedings}, volume 363 of {\em Lecture Notes in Computer Science}, pages 108--118. Springer, 1989.
\newblock \href {https://doi.org/10.1007/3-540-51237-3\_10} {\path{doi:10.1007/3-540-51237-3\_10}}.

\bibitem{hermelin2021new}
Danny Hermelin, Shlomo Karhi, Michael Pinedo, and Dvir Shabtay.
\newblock New algorithms for minimizing the weighted number of tardy jobs on a single machine.
\newblock {\em Annals of Operations Research}, 298:271--287, 2021.
\newblock \href {https://doi.org/10.1007/s10479-018-2852-9} {\path{doi:10.1007/s10479-018-2852-9}}.

\bibitem{HermelinKSWW15}
Danny Hermelin, Stefan Kratsch, Karolina Soltys, Magnus Wahlstr{\"{o}}m, and Xi~Wu.
\newblock {A Completeness Theory for Polynomial (Turing) Kernelization}.
\newblock {\em Algorithmica}, 71(3):702--730, 2015.
\newblock \href {https://doi.org/10.1007/S00453-014-9910-8} {\path{doi:10.1007/S00453-014-9910-8}}.

\bibitem{Holm2001}
Jacob Holm, Kristian de~Lichtenberg, and Mikkel Thorup.
\newblock Poly-logarithmic deterministic fully-dynamic algorithms for connectivity, minimum spanning tree, 2-edge, and biconnectivity.
\newblock {\em J. ACM}, 48(4):723–760, July 2001.
\newblock \href {https://doi.org/10.1145/502090.502095} {\path{doi:10.1145/502090.502095}}.

\bibitem{JansenB13}
Bart M.~P. Jansen and Hans~L. Bodlaender.
\newblock Vertex cover kernelization revisited - upper and lower bounds for a refined parameter.
\newblock {\em Theory Comput. Syst.}, 53(2):263--299, 2013.
\newblock \href {https://doi.org/10.1007/S00224-012-9393-4} {\path{doi:10.1007/S00224-012-9393-4}}.

\bibitem{JansenK15}
Bart M.~P. Jansen and Stefan Kratsch.
\newblock A structural approach to kernels for ilps: Treewidth and total unimodularity.
\newblock In Nikhil Bansal and Irene Finocchi, editors, {\em Algorithms - {ESA} 2015 - 23rd Annual European Symposium, Patras, Greece, September 14-16, 2015, Proceedings}, volume 9294 of {\em Lecture Notes in Computer Science}, pages 779--791. Springer, 2015.
\newblock \href {https://doi.org/10.1007/978-3-662-48350-3\_65} {\path{doi:10.1007/978-3-662-48350-3\_65}}.

\bibitem{JansenPL21}
Bart M.~P. Jansen, Marcin Pilipczuk, and Erik~Jan van Leeuwen.
\newblock A deterministic polynomial kernel for odd cycle transversal and vertex multiway cut in planar graphs.
\newblock {\em {SIAM} J. Discret. Math.}, 35(4):2387--2429, 2021.
\newblock \href {https://doi.org/10.1137/20M1353782} {\path{doi:10.1137/20M1353782}}.

\bibitem{Lassota20}
Klaus Jansen, Alexandra Lassota, and Lars Rohwedder.
\newblock Near-linear time algorithm for n-fold {ILPs} via color coding.
\newblock {\em SIAM Journal on Discrete Mathematics}, 34(4):2282--2299, 2020.
\newblock \href {https://doi.org/10.1137/19M1303873} {\path{doi:10.1137/19M1303873}}.

\bibitem{JansenR23}
Klaus Jansen and Lars Rohwedder.
\newblock {On Integer Programming and Convolution}.
\newblock In Avrim Blum, editor, {\em 10th Innovations in Theoretical Computer Science Conference (ITCS 2019)}, volume 124 of {\em Leibniz International Proceedings in Informatics (LIPIcs)}, pages 43:1--43:17, Dagstuhl, Germany, 2019. Schloss Dagstuhl -- Leibniz-Zentrum f{\"u}r Informatik.
\newblock \href {https://doi.org/10.4230/LIPIcs.ITCS.2019.43} {\path{doi:10.4230/LIPIcs.ITCS.2019.43}}.

\bibitem{kannan1987minkowski}
Ravi Kannan.
\newblock Minkowski's convex body theorem and integer programming.
\newblock {\em Mathematics of operations research}, 12(3):415--440, 1987.
\newblock \href {https://doi.org/10.1287/moor.12.3.415} {\path{doi:10.1287/moor.12.3.415}}.

\bibitem{pace21}
Leon Kellerhals, Tomohiro Koana, Andr\'{e} Nichterlein, and Philipp Zschoche.
\newblock {The PACE 2021 Parameterized Algorithms and Computational Experiments Challenge: Cluster Editing}.
\newblock In Petr~A. Golovach and Meirav Zehavi, editors, {\em 16th International Symposium on Parameterized and Exact Computation (IPEC 2021)}, volume 214 of {\em Leibniz International Proceedings in Informatics (LIPIcs)}, pages 26:1--26:18, Dagstuhl, Germany, 2021. Schloss Dagstuhl -- Leibniz-Zentrum f{\"u}r Informatik.
\newblock \href {https://doi.org/10.4230/LIPIcs.IPEC.2021.26} {\path{doi:10.4230/LIPIcs.IPEC.2021.26}}.

\bibitem{pace24}
Philipp Kindermann, Fabian Klute, and Soeren Terziadis.
\newblock {The PACE 2024 Parameterized Algorithms and Computational Experiments Challenge: One-Sided Crossing Minimization}.
\newblock In \'{E}douard Bonnet and Pawe{\l} Rz\k{a}\.{z}ewski, editors, {\em 19th International Symposium on Parameterized and Exact Computation (IPEC 2024)}, volume 321 of {\em Leibniz International Proceedings in Informatics (LIPIcs)}, pages 26:1--26:20, Dagstuhl, Germany, 2024. Schloss Dagstuhl -- Leibniz-Zentrum f{\"u}r Informatik.
\newblock \href {https://doi.org/10.4230/LIPIcs.IPEC.2024.26} {\path{doi:10.4230/LIPIcs.IPEC.2024.26}}.

\bibitem{knop2018scheduling}
Du{\v{s}}an Knop and Martin Kouteck{\`y}.
\newblock Scheduling meets n-fold integer programming.
\newblock {\em Journal of Scheduling}, 21:493--503, 2018.
\newblock \href {https://doi.org/10.1007/s10951-017-0550-0} {\path{doi:10.1007/s10951-017-0550-0}}.

\bibitem{Kratsch13}
Stefan Kratsch.
\newblock On polynomial kernels for integer linear programs: Covering, packing and feasibility.
\newblock In Hans~L. Bodlaender and Giuseppe~F. Italiano, editors, {\em Algorithms -- ESA 2013}, pages 647--658, Berlin, Heidelberg, 2013. Springer Berlin Heidelberg.
\newblock \href {https://doi.org/10.1007/978-3-642-40450-4_55} {\path{doi:10.1007/978-3-642-40450-4_55}}.

\bibitem{kratsch2018randomized}
Stefan Kratsch.
\newblock A randomized polynomial kernelization for vertex cover with a smaller parameter.
\newblock {\em SIAM Journal on Discrete Mathematics}, 32(3):1806--1839, 2018.
\newblock \href {https://doi.org/10.1137/16M1104585} {\path{doi:10.1137/16M1104585}}.

\bibitem{KratschMMPS21}
Stefan Kratsch, Tom{\'{a}}s Masar{\'{\i}}k, Irene Muzi, Marcin Pilipczuk, and Manuel Sorge.
\newblock Optimal discretization is fixed-parameter tractable.
\newblock In D{\'{a}}niel Marx, editor, {\em Proceedings of the 2021 {ACM-SIAM} Symposium on Discrete Algorithms, {SODA} 2021, Virtual Conference, January 10 - 13, 2021}, pages 1702--1719. {SIAM}, 2021.
\newblock \href {https://doi.org/10.1137/1.9781611976465.103} {\path{doi:10.1137/1.9781611976465.103}}.

\bibitem{Kratsch14}
Stefan Kratsch and Vuong~Anh Quyen.
\newblock On kernels for covering and packing ilps with small coefficients.
\newblock In Marek Cygan and Pinar Heggernes, editors, {\em Parameterized and Exact Computation}, pages 307--318, Cham, 2014. Springer International Publishing.
\newblock \href {https://doi.org/10.1007/978-3-319-13524-3_26} {\path{doi:10.1007/978-3-319-13524-3_26}}.

\bibitem{KratschW20}
Stefan Kratsch and Magnus Wahlstr{\"{o}}m.
\newblock Representative sets and irrelevant vertices: New tools for kernelization.
\newblock {\em J. {ACM}}, 67(3):16:1--16:50, 2020.
\newblock \href {https://doi.org/10.1145/3390887} {\path{doi:10.1145/3390887}}.

\bibitem{LongS22}
Yaowei Long and Thatchaphol Saranurak.
\newblock Near-optimal deterministic vertex-failure connectivity oracles.
\newblock In {\em 63rd {IEEE} Annual Symposium on Foundations of Computer Science, {FOCS} 2022, Denver, CO, USA, October 31 - November 3, 2022}, pages 1002--1010. {IEEE}, 2022.
\newblock \href {https://doi.org/10.1109/FOCS54457.2022.00098} {\path{doi:10.1109/FOCS54457.2022.00098}}.

\bibitem{lueker1975two}
George~S Lueker.
\newblock {\em Two NP-complete problems in nonnegative integer programming}.
\newblock Princeton University. Department of Electrical Engineering, 1975.

\bibitem{MARX2006394}
Dániel Marx.
\newblock Parameterized graph separation problems.
\newblock {\em Theoretical Computer Science}, 351(3):394--406, 2006.
\newblock Parameterized and Exact Computation.
\newblock \href {https://doi.org/10.1016/j.tcs.2005.10.007} {\path{doi:10.1016/j.tcs.2005.10.007}}.

\bibitem{MehlhornSU97}
Kurt Mehlhorn, R.~Sundar, and Christian Uhrig.
\newblock Maintaining dynamic sequences under equality tests in polylogarithmic time.
\newblock {\em Algorithmica}, 17(2):183--198, 1997.
\newblock \href {https://doi.org/10.1007/BF02522825} {\path{doi:10.1007/BF02522825}}.

\bibitem{Miller60}
C.~E. Miller, A.~W. Tucker, and R.~A. Zemlin.
\newblock Integer programming formulation of traveling salesman problems.
\newblock {\em J. ACM}, 7(4):326–329, October 1960.
\newblock \href {https://doi.org/10.1145/321043.321046} {\path{doi:10.1145/321043.321046}}.

\bibitem{Nemhauser88}
George Nemhauser and Laurence Wolsey.
\newblock {\em The Scope of Integer and Combinatorial Optimization}, chapter I.1, pages 1--26.
\newblock John Wiley and Sons, Ltd, 1988.
\newblock \href {https://doi.org/10.1002/9781118627372.ch1} {\path{doi:10.1002/9781118627372.ch1}}.

\bibitem{PilipczukWrochna18}
Michał Pilipczuk and Marcin Wrochna.
\newblock On space efficiency of algorithms working on structural decompositions of graphs.
\newblock {\em {ACM} Trans. Comput. Theory}, 9(4):18:1--18:36, 2018.
\newblock \href {https://doi.org/10.1145/3154856} {\path{doi:10.1145/3154856}}.

\bibitem{Rohwedder25}
Lars Rohwedder and Karol W\k{e}grzycki.
\newblock {Fine-Grained Equivalence for Problems Related to Integer Linear Programming}.
\newblock In Raghu Meka, editor, {\em 16th Innovations in Theoretical Computer Science Conference (ITCS 2025)}, volume 325 of {\em Leibniz International Proceedings in Informatics (LIPIcs)}, pages 83:1--83:18, Dagstuhl, Germany, 2025. Schloss Dagstuhl -- Leibniz-Zentrum f{\"u}r Informatik.
\newblock \href {https://doi.org/10.4230/LIPIcs.ITCS.2025.83} {\path{doi:10.4230/LIPIcs.ITCS.2025.83}}.

\bibitem{Thomasse10}
St\'{e}phan Thomass\'{e}.
\newblock A 4k2 kernel for feedback vertex set.
\newblock {\em ACM Trans. Algorithms}, 6(2), April 2010.
\newblock \href {https://doi.org/10.1145/1721837.1721848} {\path{doi:10.1145/1721837.1721848}}.

\bibitem{Wahlstrom22}
Magnus Wahlstr{\"{o}}m.
\newblock Quasipolynomial multicut-mimicking networks and kernels for multiway cut problems.
\newblock {\em {ACM} Trans. Algorithms}, 18(2):15:1--15:19, 2022.
\newblock \href {https://doi.org/10.1145/3501304} {\path{doi:10.1145/3501304}}.

\bibitem{williams2002linear}
Justin~C Williams.
\newblock A linear-size zero—one programming model for the minimum spanning tree problem in planar graphs.
\newblock {\em Networks: An International Journal}, 39(1):53--60, 2002.
\newblock \href {https://doi.org/10.1002/net.10010} {\path{doi:10.1002/net.10010}}.

\bibitem{Wlodarczyk2024}
Michał Włodarczyk.
\newblock Does subset sum admit short proofs?
\newblock In Juli{\'{a}}n Mestre and Anthony Wirth, editors, {\em 35th International Symposium on Algorithms and Computation, {ISAAC} 2024, December 8-11, 2024, Sydney, Australia}, volume 322 of {\em LIPIcs}, pages 58:1--58:22. Schloss Dagstuhl - Leibniz-Zentrum f{\"{u}}r Informatik, 2024.
\newblock \href {https://doi.org/10.4230/LIPICS.ISAAC.2024.58} {\path{doi:10.4230/LIPICS.ISAAC.2024.58}}.

\bibitem{wlodarczyk25}
Michał Włodarczyk.
\newblock Designing compact ilps via fast witness verification.
\newblock In Akanksha Agrawal and Erik~Jan van Leeuwen, editors, {\em 20th International Symposium on Parameterized and Exact Computation, {IPEC} 2025, Warsaw, Poland, September 17-19, 2025}, volume 358 of {\em LIPIcs}, pages 16:1--16:18. Schloss Dagstuhl - Leibniz-Zentrum f{\"{u}}r Informatik, 2025.
\newblock \href {https://doi.org/10.4230/LIPICS.IPEC.2025.16} {\path{doi:10.4230/LIPICS.IPEC.2025.16}}.

\bibitem{github}
Michał Włodarczyk.
\newblock {Experiments with integer programs for Vertex Cover}, 2026.
\newblock URL: \url{https://github.com/w10d/vertex-cover/tree/main}.

\end{thebibliography}


\end{document}